\documentclass[journal]{IEEEtran}
\usepackage{amsmath,amssymb,amsfonts,amsthm}
\usepackage{cleveref}
\usepackage{mathrsfs}
\usepackage{xcolor}

\usepackage{cite}
\newcommand{\tr}{{{\mathsf T}}}
\newcommand{\Lin}[1]{\textcolor{blue}{[#1]}}

\usepackage{booktabs}
\usepackage{caption}
\usepackage{graphicx}  
\usepackage{caption}
\usepackage{subfigure}
\usepackage{url}

\newtheorem{assumption}{Assumption}
\newtheorem{lemma}{Lemma}
\newtheorem{remark}{Remark}
\newtheorem{theorem}{Theorem}
\newtheorem{problem}{Problem}

\newtheorem{proposition}{Proposition}

\begin{document}

\title{On the Optimization Landscape of Observer-based Dynamic Linear Quadratic Control}

\author{Jingliang Duan, Jie Li, Yinsong Ma, Liye Tang, Guofa Li, Liping Zhang, Shengbo Eben Li, Lin Zhao
\thanks{\textit{(Corresponding author: Jie Li.)}}
\thanks{Jingliang Duan is with the School of Mechanical Engineering, University of Science and Technology Beijing, Beijing 100083, China. 
{\tt\small Email: duanjl@ustb.edu.cn}.}
\thanks{Jie Li and Guofa Li are with the College of Mechanical and Vehicle Engineering, Chongqing University, Chongqing 400044, China.
{\tt\small Email: jieli@cqu.edu.cn; liguofa@cqu.edu.cn}.}
\thanks{Yinsong Ma is with the Laboratory for Computational Sensing and Robotics, Johns Hopkins University, Baltimore, MD 21218, USA.
{\tt\small Email: yma71@jhu.edu}.}
\thanks{Liye Tang and Shengbo Eben Li are with the School of Vehicle and Mobility, Tsinghua University, Beijing 100084, China. 
{\tt\small Email: tly20@tsinghua.org.cn; lishbo@tsinghua.edu.cn}.}
\thanks{Liping Zhang is with the Department of Mathematical Sciences, Tsinghua University, Beijing 100084, China. 
{\tt\small Email: lipingzhang@tsinghua.edu.cn}.}
\thanks{Lin Zhao is with the Department of Electrical and Computer Engineering, National University of Singapore, Singapore 119077.
{\tt\small Email: elezhli@nus.edu.sg}.}
}

\maketitle

\begin{abstract}
Understanding the optimization landscape of linear quadratic regulation (LQR) problems is fundamental to the design of efficient reinforcement learning solutions. Recent work has made significant progress in characterizing the landscape of static output-feedback control and linear quadratic Gaussian (LQG) control. For LQG, much of the analysis leverages the separation principle, which allows the controller and estimator to be designed independently. However, this simplification breaks down when the gradients with respect to the estimator and controller parameters are inherently coupled, leading to a more intricate analysis. This paper investigates the optimization landscape of observer-based dynamic output-feedback control of LQR problems. We derive the optimal observer-controller pair in settings where transient quadratic performance cannot be neglected. Our analysis reveals that, in general, the combination of the standard LQR controller and the observer that minimizes the trace of the accumulated estimation error covariance does not correspond to a stationary point of the overall closed-loop performance objective. Moreover, we derive a pair of discrete-time Sylvester equations with symmetric structure, both involving the same set of matrix elements, that characterize the stationary point of the observer-based dynamic LQR problem. These equations offer analytical insight into the structure of the optimality conditions and provide a foundation for developing numerical policy gradient methods aimed at learning complex controllers that rely on reconstructed state information.
\end{abstract}


\IEEEpeerreviewmaketitle

\section{Introduction}
\IEEEPARstart{T}{he} linear quadratic regulation (LQR) problem serves as a fundamental research topic within the realm of control theory~\cite{lewis2012optimal}, whose optimal controller has a static state-feedback form. During the past few years, theoretical analysis work targeting the LQR control has made substantial progress, laying the foundation for policy structures~\cite{tang2023analysis}, convergence guarantees~\cite{fazel2018global,malik2020derivative}, efficiency improvements~\cite{mohammadi2021convergence,zhao2025asynchronous}, and sample complexity~\cite{chen2024global}. However, due to the practical challenges of obtaining full state information, output-feedback control warrants further investigation as a more realistic and significant controller design approach.

Some recent works have studied the optimization landscape of static output-feedback control~\cite{duan2023optimization}. The necessary and sufficient conditions for the stability of static output-feedback (SOF) controllers are derived by bridging the error between state-feedback gain and output-feedback gain~\cite{ilka2022novel}. On the other hand, further research has delved into the optimization landscape of dynamic output-feedback linear quadratic regulation (dLQR)~\cite{sadamoto2024policy}.~\cite{duan2024optimization} proved the observable stationary point of dLQR is unique. \cite{zhao2023globally} proposed an alternative policy parameterization using the past input-output trajectory of finite length as feedback and established a global convergence guarantee. 
Robustness constraints have been incorporated in the design of dynamical controllers, where the feasible set of stabilizing controllers is proved to have at most two connected components~\cite{hu2022connectivity}. The global optimality of the standard $H_{\infty}$ robust control with non-degenerated stabilizing dynamical structures is further considered in~\cite{tang2023global}, which reveals that all Clarke stationary points are globally optimal despite the non-smoothness of $H_{\infty}$ cost function. 
For stochastic dynamics, the linear quadratic Gaussian (LQG) problem is analyzed from the perspectives of the connectivity of stabilizing controllers and the structure of stationary points~\cite{tang2023analysis}. For unknown systems, model-free learning methods have been designed for dLQR with convergence and optimality analysis~\cite{xie2025optimal}. 

Although static output-feedback LQR offers greater design flexibility, observer-based output-feedback can achieve performance \textit{closer} to that of full state-feedback control by leveraging reconstructed state information. This advantage stems from the separation principle, which enables independent pole placement for state observers and feedback controllers~\cite{lewis2012optimal}. Given that pole locations dominate the convergence rate of estimation error and the performance deviation from the state-feedback control~\cite{xie2025optimal}, such structural decoupling has important theoretical significance. 
However, for policy gradient method of solving LQG, it has been demonstrated in LQG problems that the gradients of the objective function with respect to the controller and observer are not separable~\cite{mohammadi2021lack}. 
Based on the derived gradients, policy gradient methods demonstrate efficacy in optimizing the controller-observer pair for \textbf{observer-based dynamic LQR} (\texttt{OD-LQR}) problems through dynamical system augmentation~\cite{kumar2024optimal}. However, the analysis there overlooks the transient performance and hence cannot guarantee the optimality. Moreover, existing works only focus on developing numerical solutions. They lack rigorous analysis of both the optimality conditions and the effect of the transient quadratic performance on the optimization landscape. 

In this work, we investigate the stationary point of dynamic output-feedback LQR control with a state observer, providing novel insights into the optimality conditions of observer-based LQR. Our main theoretical results are summarized as follows:
\begin{enumerate}
\item We derive an analytical expression for the standard observer gain $L^{\star}$, which minimizes the accumulated estimation variance. This gain serves as the optimal observer gain for \texttt{OD-LQR}, analogous to the optimal filter gain in the LQG setting. Unlike observer gains determined via pole assignment, which lacks transient performance guarantee, our theoretical analysis and numerical results demonstrate that $L^{\star}$ ensures control optimality under the standard state-feedback LQR gain $K^{\star}$ \Lin{\Cref{theorem.stationary_point}}.
\item Using policy gradient expressions, we analyze the optimization landscape of \texttt{OD-LQR}. Our results demonstrate that while the standard LQR gain $K^{\star}$ is commonly used in observer-based designs, it is generally only suboptimal in achieving the minimal quadratic performance when combined with the standard observer gain $L^{\star}$ in the \texttt{OD-LQR} problem. It becomes optimal when the initial cross-correlation between observation error and state observation vanishes \Lin{\Cref{proposition.stationary_failed}}.
\item We derive the set of stationary points of the \texttt{OD-LQR} problem by solving the first-order optimality conditions with the coupled gradients. We show that the solution set satisfies a pair of discrete-time Sylvester equations with symmetric structure. Interestingly, it is shown that when the aforementioned cross-correlation vanishes, the set collapses into the standard controller-observer pair $(L^{\star}, K^{\star})$ (see \Lin{\Cref{theorem.sufficient_condition}} and \Lin{\Cref{proposition.stationary_degeneration}}).
\end{enumerate}

To the best of our knowledge, this work represents the first systematic investigation into the optimality of observer-based dynamic LQR. Our findings offer new theoretical insights for designing policy gradient methods for \texttt{OD-LQR} problems, where the stationary point does not necessarily adhere to the separation principle when optimizing the objective function.

The rest of the work is organized as follows. Section~\ref{sec.problem} presents the formulation of \texttt{OD-LQR} problems. Section~\ref{sec.gradient} presents our main results on policy gradients and stationary points. Numerical experiments are shown in Section~\ref{sec.experiment} and Section~\ref{sec.conclusion} concludes this work. 

\textbf{Notations.} For $X \in \mathbb{R}^{n \times n}$, we use $\rho(X)$ to denote its spectral radius. The notations $\mathbb{S}^n_+$ and $X \succeq 0$ (respectively, $\mathbb{S}^n_{++}$ and $X \succ 0$) denotes the set of symmetric $n \times n$ positive semi-definite (respectively, positive definite) matrices. We employ the superscript $\star$ to denote the standard optimal controller $K^\star$ in LQR and the standard observer $L^\star$ that minimizes the trace of the accumulated state estimation variance, and use $\ddagger$ to denote the stationary points $(K^\ddagger, L^\ddagger)$ of \texttt{OD-LQR} problems. 

\section{Problem Statement}
\label{sec.problem}
\subsection{Linear Quadratic Control}
Consider a discrete-time linear time-invariant (LTI) system 
\begin{equation}  
\label{eq.state_function}
\begin{aligned}
x_{t+1} &= Ax_t+Bu_t,\\
y_t &= Cx_t,
\end{aligned}
\end{equation}
where $x_t \in \mathbb{R}^n$, $u_t \in \mathbb{R}^m$, and $y_t \in \mathbb{R}^d$ denote the states, control inputs, and observation outputs separately, and $A \in \mathbb{R}^{n\times n}$, $B \in \mathbb{R}^{n\times m}$, $C \in \mathbb{R}^{d\times n}$ are known dynamics. The linear quadratic control aims to find control input $u_t$ to minimize the cumulative quadratic utilities: 
\begin{equation}
\label{eq.objective}
\min_{u_t} \ \mathbb{E}_{x_0} \left[\sum_{t=0}^{\infty}\left(x_t^{\tr} Q x_t + u_t^{\tr} R u_t\right) \right].
\end{equation}
The initial state distribution for $x_0$ is assumed to satisfy that $\mathbb{E}_{x_0}[x_0 x_0^{\tr}] \succ 0$, a covariance condition commonly adopted in data-driven control~\cite{fazel2018global,  lee2018primal,malik2020derivative,hu2022towards} that parallels the persistent excitation requirement~\cite{de2019formulas}. For the above system~\eqref{eq.state_function} and problem~\eqref{eq.objective}, we make the following assumptions: 
\begin{assumption}
\label{assumption.control_observe}
$Q \in \mathbb{S}_+^{n}$, $R \in \mathbb{S}_{++}^{m}$, $(A,B)$ is controllable, and $(C, A)$ and $(Q^{\frac{1}{2}},A)$ are observable. We maintain generality by considering $C$ to have rows of full rank. 
\end{assumption}

When $C$ coincides with the identity matrix $I_n$, the optimal controller takes the static form $u_t = K x_t$, where $K$ can be derived from the associated Riccati equation. 
Under \Cref{assumption.control_observe}, observer-based dynamic controllers with stabilization guarantees enjoy computational tractability by the separation principle~\cite{lewis2012optimal}, which will be presented in  Section~\ref{subsec.observer_design}. 

\begin{remark}
While a continuous-time variant could be conceptualized, extending it to policy optimization would exacerbate the curse of dimensionality due to function approximation in continuous state-spaces. In contrast, our discrete-time formulation provides exact analytical gradients to circumvent this issue. Furthermore, the mathematical structure changes fundamentally. The continuous-time optimality conditions would manifest as coupled algebraic Riccati equations rather than the symmetric Sylvester equations derived in this work. 
\end{remark}

\subsection{Observer Design}
\label{subsec.observer_design}
Consider the standard observer-based dynamic controller
\begin{equation}
\label{eq:observer-based-controller}
\begin{aligned}
    \xi_{t+1} &= \mathcal{A}_{K,L} \xi_t + Ly_{t}, \\
    u_t &= -K \xi_{t},
\end{aligned}
\end{equation}
with $\mathcal{A}_{K,L} := A - BK - LC$, where $\xi_t \in \mathbb{R}^n$ is the state observation and also the internal state of dynamic controller, $K \in \mathbb{R}^{m \times n}$ is the undetermined controller gain, and $L \in \mathbb{R}^{n \times d}$ is the observer gain to be solved. 

According to the separation principle, the controller~\eqref{eq:observer-based-controller} is stabilizing if and only if $K$ and $L$ are stabilizing‌ gains~\cite{zhang2020robust}.
We define the stabilizing set of $K$ and $L$ as
\begin{equation}
\label{eq:stabilizing-K}
\begin{aligned}
\mathbb{K}:&= \{K\in \mathbb{R}^{m\times n}:\rho(A-BK)<1\},\\
\mathbb{L}:&= \{L\in \mathbb{R}^{n\times d}:\rho(A-LC)<1\}.
\end{aligned}    
\end{equation}
A widely-used selection is to choose $K$ as the standard state-feedback LQR gain $K^{\star}$ and find a stabilizing observer gain using pole assignment method such that $A-LC$ converges faster than $A-BK$. Although we can obtain a stabilizing observer gain, its optimality with respect to the accumulated quadratic utilities~\eqref{eq.objective} has not been fully considered. Section~\ref{subsec.cost_function} will analyze the quadratic cost from the perspective of the Lyapunov equation of the augmented system. 

\subsection{Cost Function in the \texttt{OD-LQR} Problem}
\label{subsec.cost_function}
The closed-loop dynamics of the LTI system under the dynamic controller is
\begin{equation}   
\label{eq.closed-loop-system}
\begin{bmatrix}
     x_{t+1}\\
     \xi_{t+1} 
\end{bmatrix}= \begin{bmatrix}
     A & -BK\\
     LC & A-BK-LC
\end{bmatrix}\begin{bmatrix}
     x_{t}\\
     \xi_{t} 
\end{bmatrix}.
\end{equation}
For the convenience of subsequent analysis, a linear transformation $T$ is performed on the augmented state $\begin{bmatrix}x^{\tr}_{t}&\xi^{\tr}_{t}\end{bmatrix}^{\tr}$, resulting in the transformed augmented system state as follows 
\begin{equation}
\nonumber
\bar{z}_t: = \begin{bmatrix}
x_t \\
x_t - \xi_t
\end{bmatrix} = \begin{bmatrix}
     I_n & 0\\
     I_n & -I_n
\end{bmatrix}\begin{bmatrix}
x_t \\
\xi_t
\end{bmatrix} = T \begin{bmatrix}
x_t \\
\xi_t
\end{bmatrix}, 
\end{equation}
whose initial distribution is denoted as $\mathcal{B}$. Therefore, the dynamics of the augmented system is expressed as 
\begin{equation}   
\nonumber
\begin{bmatrix}
     x_{t+1}\\
     x_{t+1} - \xi_{t+1} 
\end{bmatrix}= \begin{bmatrix}
     A - BK & BK\\
     0 & A - LC
\end{bmatrix}\begin{bmatrix}
     x_{t}\\
     x_{t} - \xi_{t} 
\end{bmatrix},
\end{equation}
where $x_{t} - \xi_{t}$ is the observation error. 
We further denote $$\bar{A}:=\begin{bmatrix}
     A & 0\\
     0 & A
\end{bmatrix}, \;\hat{B}:=\begin{bmatrix}
     B\\
     0
\end{bmatrix},\; \hat{C}:=\begin{bmatrix}0&-C\end{bmatrix},$$ then the close-loop system~\eqref{eq.closed-loop-system} is denoted as 
\begin{equation}
\label{eq.closed-loop-system_2}
\bar{z}_{t+1} = \hat{\mathcal{A}}_{K,L}\bar{z}_{t},
\end{equation}
where $\hat{\mathcal{A}}_{K,L}:=\bar{A} - \hat{B} K \bar{F} + \hat{F}^{\tr} L \hat{C}$, $\bar{F}:=\begin{bmatrix}I_n&-I_n\end{bmatrix}$ and $\hat{F}:=\begin{bmatrix}0&I_n\end{bmatrix}$.

To link our formulation with the general framework of policy optimization, we first define the value function at time~$t$ for the augmented state $\bar{z}_t$ under fixed gains $(K, L)$ as: 
\begin{equation}
\nonumber
V_t(\bar{z}_t) = \bar{z}_t^{\tr} S_t \bar{z}_t,
\end{equation}
where $S_t$ is time-dependent. According to the principle of dynamic programming, the Bellman equation is 
\begin{equation}
\nonumber
S_t = \hat{Q} + \hat{\mathcal{A}}_{K,L}^{\tr} S_{t+1} \hat{\mathcal{A}}_{K,L}.
\end{equation}
As the system evolves and approaches the steady state under stabilizing gains $(K, L)$, the value function becomes time-invariant. This implies the convergence of the value matrix: 
\begin{equation}
\nonumber
\lim_{t \to \infty} S_{t+1} = S_t = S_{K,L}.
\end{equation}
Therefore, the quadratic cost functional is formulated as
\begin{equation}
\nonumber
V_{K,L}(\bar{z}_t): = \bar{z}_t^{\tr} S_{K,L}\bar{z}_t,
\end{equation}
where $S_{K,L} \in \mathbb{S}_{+}^{2n}$ satisfies a Lyapunov equation~\eqref{eq.lyapunov_equation}, as shown in \Cref{lemma.OD-LQR_cost}. Let the cumulative state correlation driven by a stabilizing gain $K \in \mathbb{K}$ be formally defined as
\begin{equation}
\nonumber
\Omega_{K,L}:=\mathbb{E}_{\bar{z}_0\sim \mathcal{B}} \left[\sum_{t=0}^{\infty}\bar{z}_t \bar{z}_t^{\tr}\right].
\end{equation}
For each $K \in \mathbb{K}$ and $L \in \mathbb{L}$, there is an associated $S_{K,L}$ and $\Omega_{K,L}$. These matrices provide a convenient way to express the \texttt{OD-LQR} cost function, as summarized in \Cref{lemma.OD-LQR_cost}.

\begin{lemma}\cite[Lemma 1]{duan2024optimization}
\label{lemma.OD-LQR_cost}
For any stabilizing controller $K \in \mathbb{K}$ and observer $L \in \mathbb{L}$, the \texttt{OD-LQR} cost function is 
\begin{equation}   
\label{eq.cost_in_P}
J(K,L) = {\rm Tr}(S_{K,L}Y) = {\rm Tr}(\Omega_{K,L}\hat{Q}), 
\end{equation}
where $S_{K,L}$ and $\Omega_{K,L}$ constitute the unique positive semi-definite solutions to specified Lyapunov equations
\begin{subequations}
\begin{align} 
\label{eq.lyapunov_equation}
S_{K,L} &= \hat{Q} + \hat{\mathcal{A}}_{K,L}^{\tr} S_{K,L}\hat{\mathcal{A}}_{K,L},  \\
\Omega_{K,L} &= Y + \hat{\mathcal{A}}_{K,L}\Omega_{K,L}\hat{\mathcal{A}}_{K,L}^{\tr}, \label{eq.lyapunov_equation_sigma}
\end{align}
\end{subequations}
where $Y := \mathbb{E}_{\bar{z}_0 \sim \mathcal{B}} [\bar{z}_0\bar{z}_0^{\tr}]$ is the initial state correlation, and 
\begin{equation}
\nonumber
\hat{Q} = \begin{bmatrix}Q & 0\\0 & 0\end{bmatrix} + \bar{F}^{\tr} K^{\tr} R K \bar{F} = \begin{bmatrix}Q + K^{\tr} R K & - K^{\tr} R K\\- K^{\tr} R K & K^{\tr} R K\end{bmatrix}.
\end{equation}
\end{lemma}

To facilitate subsequent analysis involving the inversion of positive definite matrices, some requirements are given: 

\begin{assumption}
\label{assumption.independent}
The initial correlation $Y$ is independent of the controller and observer gains, and is strictly positive definite. 
\end{assumption}

The above assumptions apply in the sequel. Different from \texttt{OD-LQR}, LQG minimizes a limiting average cost, yields a correlation representing the steady-state covariance, which inherently depends on controller and observer gains. Based on the derived cost function, we formulate the \texttt{OD-LQR} problem: 
\begin{problem}[Optimization for \texttt{OD-LQR}]
\label{problem.OD-LQR}
Suppose $\bar{z}_0$ follows $\mathcal{B}$. The optimal \texttt{OD-LQR} control is formulated as: 
\begin{equation}  
\nonumber
\begin{aligned}
\min_{K,L} \quad  J(K,L) &= \min_{K,L} \;\; \mathbb{E}_{\bar{z}_0 \sim \mathcal{B}}\left[\sum_{t=0}^{\infty}\left(x_t^{\tr} Q x_t + u_t^{\tr} R u_t\right) \right] \\
\text{\rm subject to} \quad &\eqref{eq.state_function}, \ \eqref{eq:observer-based-controller}, \ K\in \mathbb{K}, \ L\in \mathbb{L},
\end{aligned}
\end{equation}
where the cost function $J(K,L)$ is calculated by~\eqref{eq.cost_in_P}, and stabilizing sets $\mathbb{K}$ and $\mathbb{L}$ are defined in~\eqref{eq:stabilizing-K}. 
\end{problem}

The following gives the connectivity of stabilizing sets~\cite{bu2020topological}. 

\begin{lemma}[Domain Connectivity]
The sets of stabilizing controller gain $\mathbb{K}$ and observer gain $\mathbb{L}$ are path-connected. 
\end{lemma}

In this work, the closed-form gradient expressions for $J(K,L)$ will be established to characterize the optimization geometry of \Cref{problem.OD-LQR}. Prior to delving into theoretical analysis, Section~\ref{subsec.block_Lyapunov} presents the block-wise Lyapunov equations and summarizes the associated Lyapunov stability theorems. 

\subsection{Block-wise Lyapunov Equations}
\label{subsec.block_Lyapunov}
The block-structured Lyapunov equations in~\eqref{eq.lyapunov_equation} and~\eqref{eq.lyapunov_equation_sigma} serve as foundational analytical tools in this work, where $S_{K,L}$ can be divided into the following blocks: 
\begin{equation} \label{eq:Pk-partition}
S_{K,L}=\begin{bmatrix}
   S_{11}  &  S_{12}\\
   S_{12}^{\tr}  & S_{22} 
\end{bmatrix}. 
\end{equation}
For notational conciseness, subscripts $K$ and $L$ in the submatrices of $S_{K,L}$ and $\Omega_{K,L}$ will be omitted when $K$ and $L$ dependencies are clear in the context. From~\eqref{eq.lyapunov_equation}, one has 
\begin{subequations}
\label{eq.block_lyapunov_P}
\begin{align}
&\begin{aligned}
\label{eq.block_lyapunov_P11}
S_{11} &= Q + K^{\tr} R K + \left(A - B K\right)^{\tr} S_{11} \left(A - B K\right),
\end{aligned}\\
&\begin{aligned}
\label{eq.block_lyapunov_P12}
S_{12} &= - K^{\tr} R K + \left(A - B K\right)^{\tr} S_{11} B K \\
&\quad \ + \left(A - B K\right)^{\tr} S_{12} \left(A - L C\right),
\end{aligned}\\
&\begin{aligned}
\label{eq.block_lyapunov_P22}
S_{22} &= K^{\tr} R K + K^{\tr} B^{\tr} S_{11} B K + K^{\tr} B^{\tr} S_{12} \left(A - L C\right) \\
&\quad \ + \left(A - L C\right)^{\tr} S_{12}^{\tr} B K + \left(A - L C\right)^{\tr} S_{22} \left(A - L C\right).
\end{aligned}
\end{align}
\end{subequations}

Similarly, let 
\begin{equation} \label{eq:Sigma-X-partition}
\Omega_{K,L}=\begin{bmatrix}
   \Omega_{11}  &  \Omega_{12}\\
   \Omega_{12}^{\tr}  & \Omega_{22} 
\end{bmatrix}, \;\; Y=\begin{bmatrix}
   Y_{11}  &  Y_{12}\\
   Y_{12}^{\tr}  & Y_{22} 
\end{bmatrix},
\end{equation}
where $Y_{22}$ denotes the initial correlation of observation error, and $Y_{12}^{\tr}$ indicates the initial cross-correlation between observation error and system state. From~\eqref{eq.lyapunov_equation_sigma}, we get
\begin{subequations}
\label{eq.block_lyapunov_Sigma}
\begin{align}
&\begin{aligned}
\label{eq.block_lyapunov_Sigma11}
\Omega_{11} &= Y_{11} + \left(A - B K\right) \Omega_{11} \left(A - B K\right)^{\tr} \\
&\quad \ + \left(A - B K\right) \Omega_{12} K^{\tr} B^{\tr} \\
&\quad \ + B K \Omega_{12}^{\tr} \left(A - B K\right)^{\tr} + B K \Omega_{22} K^{\tr} B^{\tr},
\end{aligned}\\
&\begin{aligned}
\label{eq.block_lyapunov_Sigma12}
\Omega_{12} &= Y_{12} + \left(A - B K\right) \Omega_{12} \left(A - L C\right)^{\tr} \\
& \quad \ + B K \Omega_{22} \left(A - L C\right)^{\tr}, 
\end{aligned}\\
&\begin{aligned}
\label{eq.block_lyapunov_Sigma22}
\Omega_{22} &= Y_{22} + \left(A - L C\right) \Omega_{22} \left(A - L C\right)^{\tr}.
\end{aligned}
\end{align}
\end{subequations}

Lyapunov stability theorems will be employed as analytical fundamentals. Essential propositions are compiled as follows: 
\begin{lemma}[Lyapunov Stability Theorems~\cite{gu2012discrete,lee2018primal}]
\label{lemma.Lyapunov_stability}
\ 
\begin{enumerate}
\item[a)] If $\rho(A) <1$ and $Q \in \mathbb{S}_{+}^n$, the Lyapunov equation $P = Q + A^{\tr} P A$ has a unique solution $P \in \mathbb{S}_{+}^n$.
\item[b)] Let $Q \in \mathbb{S}_{++}^n$. $\rho(A) <1$ if and only if there exists a unique $P \in \mathbb{S}_{++}^n$ such that $P = Q + A^{\tr} P A$.
\end{enumerate}
\end{lemma}

\section{Gradients and Stationary Points}
\label{sec.gradient}
This section first establishes the closed-form gradients of the \texttt{OD-LQR} cost with respect to feedback controller gain $K$ and state observer gain $L$. Then, we investigate the stationary point where the gradients vanish and establish its relationship between the standard LQR controller and standard observer. Finally, we derive the stationary point of the \texttt{OD-LQR} problem through the Sylvester equation and explore the specific conditions under which the stationary point collapses into the standard controller-observer pair. 

\subsection{The Gradient of the \texttt{OD-LQR} Cost}
The gradient of \texttt{OD-LQR} problem is the basis for analyzing stationary points. \Cref{lemma.gradient} develops analytical expressions for the gradients of the \texttt{OD-LQR} cost function with respect to the controller and observer gains. 
\begin{lemma}[Policy Gradient]
\label{lemma.gradient}
Given any stabilizing controller gain $K \in \mathbb{K}$ and observer gain $L\in \mathbb{L}$, the gradients are 
\begin{equation}  
\label{eq.gradient}
\begin{aligned}
\nabla_{K} J(K,L) &= 2(R K \bar{F} - \hat{B}^{\tr} S_{K,L} \hat{\mathcal{A}}_{K,L})\Omega_{K,L} \bar{F}^{\tr},\\
\nabla_L J(K,L) &= 2\hat{F} S_{K,L} \hat{\mathcal{A}}_{K,L} \Omega_{K,L}\hat{C}^{\tr}. 
\end{aligned}
\end{equation}
\end{lemma}

\begin{proof}
The proof methodology parallels the state-feedback LQR framework~\cite[Lemma 1]{fazel2018global}. According to the specified Lyapunov equation~\eqref{eq.lyapunov_equation}, the cost function of $\bar{z}_0$ is 
\begin{equation}
\nonumber
\begin{aligned}
V_{K,L}(\bar{z}_0) &= \bar{z}_0^{\tr} (\hat{Q} + \hat{\mathcal{A}}_{K,L}^{\tr} S_{K,L}\hat{\mathcal{A}}_{K,L})\bar{z}_0\\
&= \bar{z}_0^{\tr} \hat{Q} \bar{z}_0 + V_{K,L}(\hat{\mathcal{A}}_{K,L}\bar{z}_0).
\end{aligned}
\end{equation}
Taking the gradient of $V_{K,L}(\bar{z}_0)$ with respect to $K$, one has
\begin{equation}
\nonumber
\resizebox{1.0\hsize}{!}{$
\begin{aligned}
\nabla_{K} V_{K,L}(\bar{z}_0) &= 2(R K \bar{F} - \hat{B}^{\tr} S_{K,L}\hat{\mathcal{A}}_{K,L}) \bar{z}_0\bar{z}_0^{\tr} \bar{F}^{\tr} + \bar{z}_1^{\tr} \nabla_{K} S_{K,L} \bar{z}_1\\
&= 2(R K \bar{F} - \hat{B}^{\tr} S_{K,L}\hat{\mathcal{A}}_{K,L}) \bar{z}_0\bar{z}_0^{\tr} \bar{F}^{\tr} + \nabla_{K} V_{K,L}(\bar{z}_1)\\
&= 2(R K \bar{F} - \hat{B}^{\tr} S_{K,L}\hat{\mathcal{A}}_{K,L}) \sum_{t=0}^{\infty}(\bar{z}_t\bar{z}_t^{\tr}) \bar{F}^{\tr},
\end{aligned}
$}
\end{equation}
where the last equation follows by recursion and the fact that $\bar{z}_{t+1} = \hat{\mathcal{A}}_{K,L}\bar{z}_t = (\bar{A} - \hat{B} K \bar{F} + \hat{F}^{\tr} L \hat{C})\bar{z}_t$. 

Similarly, we get
\begin{equation}
\nonumber
\begin{aligned}
\nabla_L V_{K,L}(\bar{z}_0) 
&= 2\hat{F} S_{K,L}\hat{\mathcal{A}}_{K,L}\bar{z}_0\bar{z}_0^{\tr} \hat{C}^{\tr} + \bar{z}_1^{\tr} \nabla_L S_{K,L} \bar{z}_1\\
&= 2\hat{F} S_{K,L}\hat{\mathcal{A}}_{K,L}\bar{z}_0\bar{z}_0^{\tr} \hat{C}^{\tr} + \nabla_L V_{K,L}(\bar{z}_1)\\
&= 2\hat{F} S_{K,L}\hat{\mathcal{A}}_{K,L}\sum_{t=0}^{\infty}(\bar{z}_t\bar{z}_t^{\tr}) \hat{C}^{\tr}.
\end{aligned}
\end{equation}
By calculating the stochastic average of gradient expressions across the initial distribution $\mathcal{B}$, we have~\eqref{eq.gradient}.
\end{proof}

The obtained gradient expression is related to the positive semi-definite solutions $S_{K,L}$ and $\Omega_{K,L}$ to the Lyapunov equations. Section~\ref{subsec.optimal_ctrl_obsv} will utilize $S_{K,L}$ and $\Omega_{K,L}$ to represent the standard LQR controller and standard observer minimizing the accumulated estimation variance, respectively. Before that, we will present the results about gradient dominance. 
\begin{lemma}[Gradient Dominance]
\label{lemma.gradient_dominance}
Assuming that the closed-loop matrix $\hat{\mathcal{A}}_{K,L}$ admits a uniform spectral norm, i.e., $\|\hat{\mathcal{A}}_{K,L}\|_2 \le \gamma < 1$, there exist local attraction radii $r_K > 0$ and $r_L > 0$, such that for all controller gain $K$ and observer gain $L$ satisfying $\|\Delta K\|_F \le r_K$ and $\|\Delta L\|_F \le r_L$, the cost difference is upper bounded by 
\begin{equation}
\nonumber
J(K,L) - J(K^{\ddagger},L^{\ddagger}) \le \frac{\|\nabla_K J(K, L)\|_F^2}{2 (\alpha_{K2} - \beta_{K2})} + \frac{\|\nabla_L J(K, L)\|_F^2}{2 (\alpha_{L2} - \beta_{L2})} , 
\end{equation}
where the coefficients associated with the controller gain $K$ and the observer gain $L$ exhibit strictly decoupled structures. 
\end{lemma}
\begin{proof}
See \Cref{sec.gradient_dominance} for the proof. 
\end{proof}

\subsection{Standard LQR Controller and Standard Observer}
\label{subsec.optimal_ctrl_obsv}
According to the classical control theory, the standard optimal state-feedback controller gain for the linear system~\eqref{eq.state_function} under the quadratic cost~\eqref{eq.objective} is 
\begin{equation}
\label{eq.standard_K}
K^{\star}:= (R + B^{\tr} {\hat{S}^{\star}} B)^{-1} B^{\tr} {\hat{S}^{\star}} A,
\end{equation}
where $\hat{S}^{\star}$ is the unique positive definite solution to
\begin{equation}
\label{eq.P_riccati}
{\hat{S}^{\star}} = Q + A^{\tr} {\hat{S}^{\star}} A - A^{\tr} {\hat{S}^{\star}} B (R + B^{\tr} {\hat{S}^{\star}} B)^{-1} B^{\tr} {\hat{S}^{\star}} A.
\end{equation}
Note that the above equation~\eqref{eq.P_riccati} is equivalent to 
\begin{equation}
\label{eq.P_Kstar}
\hat{S}^{\star} = Q + {K^{\star}}^{\tr} R K^{\star} + (A - B K^{\star})^{\tr} \hat{S}^{\star} (A - B K^{\star}),
\end{equation}
which is similar to the block-wise Lyapunov equation~\eqref{eq.block_lyapunov_P11}. 

If $(K,\mathcal{A}_{K,L})$ is observable, then the dynamic controller~\eqref{eq:observer-based-controller} is referred as \textit{observable}. The set of observable controllers under the standard LQR controller gain $K^{\star}$ is denoted as 
\begin{equation}
\nonumber
\mathbb{L}_o := \left\{ L\in \mathbb{R}^{n \times d}: (K^{\star},\mathcal{A}_{K^{\star},L}) \text{ is observable} \right\}.
\end{equation}
The following result on the unique positive definite solution of the Lyapunov equation directly follows from \Cref{lemma.Lyapunov_stability}. 
\begin{lemma}
\label{lemma.positive_P}
If $L\in \mathbb{L}\cap\mathbb{L}_o$ and $\rho\left(\mathcal{A}_{K^{\star},L}\right) < 1$, the solution $S_{K^{\star},L}$ to~\eqref{eq.lyapunov_equation} is unique and positive definite. Specifically, the blocks of the solution $S_{K^{\star},L}$ satisfy $S_{11}^{\star} = \hat{S}^{\star}$ and $S_{12}^{\star} = 0$. 
\end{lemma}
\begin{proof}
Consider the block-wise Lyapunov equation~\eqref{eq.block_lyapunov_P11} with fixed dynamics system, where $S_{11}$ only depends on $K$. In particular, for the standard LQR controller $K^{\star}$, the closed-loop system $A - B K^{\star}$ is stable. According to \Cref{lemma.Lyapunov_stability}(a), the solution to~\eqref{eq.block_lyapunov_P11} with $K^{\star}$, denoted as $S_{11}^{\star}$, is unique. Since the block-wise Lyapunov equation~\eqref{eq.block_lyapunov_P11} shares the same structure as the Lyapunov equation~\eqref{eq.P_Kstar}, we have $S_{11}^{\star} = \hat{S}^{\star}$. 

Substituting the standard LQR controller gain $K^{\star}$ into the block-wise Lyapunov equation~\eqref{eq.block_lyapunov_P12} yields 
\begin{equation}
\label{eq.P_22-P_12-op}
\begin{aligned}
S_{12}^{\star} &= \left(A - B K^{\star}\right)^{\tr} S_{12}^{\star} \left(A - L C\right). 
\end{aligned}
\end{equation}
By adapting the analysis framework of the unique solution of the Lyapunov equation in~\cite[Theorem 8.2.2]{datta2004numerical}, and noting that $\rho(A - B K^{\star}) < 1$ and $\rho(A - L C) < 1$ imply that the product of any eigenvalue pair is strictly less than one, the above Sylvester  equation has a unique solution, namely $S_{12}^{\star} = 0$. 
\end{proof}

Next, the observer gain $L^{\star}$ that minimizes the accumulated estimation variance is introduced.
Define the estimation variance at $t$ under observer $L$ as $E_t:= \mathbb{E}_{\bar{z}_0 \sim \mathcal{B}} [(x_t-\xi_t)(x_t-\xi_t)^{\tr}]$. 
It is not hard to find that $E_0 = Y_{22}$. 
Subtracting~\eqref{eq.state_function} from~\eqref{eq:observer-based-controller}, we can derive that 
\begin{equation}
\label{eq.estimation_variance_rela}
E_{t+1} = \left(A - L C\right)^{t+1} E_0 \left(A^{\tr} - C^{\tr} L^{\tr}\right)^{t+1}.
\end{equation}

From~\eqref{eq.estimation_variance_rela}, we define the accumulated state estimation variance under a stabilizing observer $L \in \mathbb{L}$ as $\hat{\Omega}_{L} := \sum_{t=0}^{\infty}E_t$. 
Therefore, $\hat{\Omega}_{L}$ satisfies the following Lyapunov equation
\begin{equation}
\label{eq.accumulated_estimation_lyapu}
\hat{\Omega}_L = E_0 + (A-LC) \hat{\Omega}_L (A-LC)^{\tr}, 
\end{equation}
which is exactly the block-wise Lyapunov equation~\eqref{eq.block_lyapunov_Sigma22}. Thus, for any stabilizing observer $L \in \mathbb{L}$, we have $\Omega_{22} = \hat{\Omega}_L$. \Cref{prob.optimal_observer} further gives the definition of the standard observer. 

\begin{problem}[Standard Observer]
\label{prob.optimal_observer}
A standard observer of system~\eqref{eq.state_function} is defined as the one that minimizes the trace of the accumulated state estimation variance $\hat{\Omega}_L$, whose gain $L^{\star}$ should be the optimal solution to 
\begin{equation}
\label{eq.problem_of_observer}
\begin{aligned}
 \min_{L\in \mathbb{L}} \quad&{\rm Tr}(\hat{\Omega}_L)\\
\text{\rm subject to} \quad &~\eqref{eq.accumulated_estimation_lyapu}.
\end{aligned}
\end{equation}
\end{problem}

The cost function ${\rm Tr}(\hat{\Omega}_L)$ denotes the sum of estimation variance, which is analogous to LQR by viewing the estimation error dynamics with $L$ as the state-feedback gain. Moreover, the constraint acts as an evaluation equation for the cost, which compresses the infinite-horizon summation into a self-consistent equation. The following proposition provides the formulation to calculate the standard observer $L^{\star}$, which follows the solution to the Riccati equation~\cite[Prop. 3]{duan2024optimization}. 

\begin{proposition}[Standard Observer Gain]
\label{proposition.optimal_observer_1}
The standard observer gain is in the form of
\begin{equation}
\label{eq.standard_L}
L^{\star} = A \hat{\Omega}_{L^\star} C^{\tr} (C \hat{\Omega}_{L^\star} C^{\tr})^{-1} \in \mathbb{L},
\end{equation}
with $\hat{\Omega}_{L^\star}$ being the unique positive definite solution to 
\begin{equation}
\label{eq.sigma_riccati}
\hat{\Omega}_{L^\star} \!=\! E_0 \!+\! A \hat{\Omega}_{L^\star} A^{\tr} \!-\! A \hat{\Omega}_{L^\star} C^{\tr} (C \hat{\Omega}_{L^\star} C^{\tr})^{-1} C \hat{\Omega}_{L^\star} A^{\tr}.
\end{equation}
\end{proposition}

\subsection{Structure of the Stationary Point}
This section will further investigate the stationary point at which the gradients vanish. We will first reveal the relationship between the stationary point, the standard LQR controller $K^{\star}$, and the standard observer $L^{\star}$. Then, we will derive the expression the stationary point and discuss the special case of the derived stationary point collapsing into the standard pair. \Cref{theorem.stationary_point} establishes the observer achieving the optimal cost function under the standard LQR controller. 
\begin{theorem}
\label{theorem.stationary_point}
Given the standard LQR controller gain $K^{\star}$, the standard observer $L^{\star}$ defined in~\eqref{eq.standard_L} is a stationary point of $J(K^{\star}, L)$. If $L^{\star} \in \mathbb{L}_o$, $L^{\star}$ is the unique observable stationary point. Otherwise, no observable stationary point exists. 
\end{theorem}

\begin{proof}
Applying the partitions defined in~\eqref{eq:Pk-partition} and~\eqref{eq:Sigma-X-partition}, the policy gradient $\nabla_L J(K,L)$ in~\eqref{eq.gradient} can be expanded as
\begin{equation}
\label{eq.gradient_L1}
\begin{aligned}
\nabla_{L} J(K,L) &= - 2 S_{12}^{\tr} (A - B K) \Omega_{12} C^{\tr} - 2 S_{12}^{\tr} B K \Omega_{22} C^{\tr} \\
&\quad \ - 2 S_{22} (A - L C) \Omega_{22} C^{\tr}.
\end{aligned}
\end{equation}
Substituting the standard LQR controller $K^{\star}$ derives that 
\begin{equation}
\label{eq.gradient_L1_K*}
\begin{aligned}
\nabla_{L} J(K^{\star},L) &= - 2 S_{22}^{\star} (A - L C) \Omega_{22} C^{\tr}.
\end{aligned}
\end{equation}

Note that $E_0 = Y_{22}$, and the Lyapunov equations~\eqref{eq.block_lyapunov_Sigma22} and~\eqref{eq.accumulated_estimation_lyapu} have the same expression. By \Cref{lemma.Lyapunov_stability}(b), we have $\Omega_{22} = \hat{\Omega}_L\in \mathbb{S}_{++}^n$ since $E_0 \in \mathbb{S}_{++}^n$. It directly follows that 
\begin{equation}
\nonumber
\resizebox{1.0\hsize}{!}{$
(A - L^\star C) \hat{\Omega}_{L\star} C^{\tr} = (A - A \hat{\Omega}_{L^\star} C^{\tr} (C \hat{\Omega}_{L^\star} C^{\tr})^{-1} C) \hat{\Omega}_{L\star} C^{\tr} = 0. 
$}
\end{equation}
This completes the proof of the first claim.

By \Cref{lemma.positive_P}, $S_{K^{\star}, L}\in \mathbb{S}_{++}^{2n}$ for any $L \in \mathbb{L}_o$, which directly leads to $S_{22}^{\star}\in \mathbb{S}_{++}^n$. Since $\Omega_{22} \in \mathbb{S}_{++}^n$ and $C$ has rows of full rank, we have $C \Omega_{22} C^{\tr}\in \mathbb{S}_{++}^n$. Therefore, by combining~\eqref{eq.accumulated_estimation_lyapu} with~\eqref{eq.sigma_riccati}, $L^{\star}$ is the unique observable stationary point. 
\end{proof}

The key to the proof is that when the controller adopts the standard LQR controller $K^{\star}$, the block matrix $S_{12}^{\star} = 0$,~\emph{i.e.}, the objective function~\eqref{eq.cost_in_P} is independent of the cross term between the state $x_t$ and the state estimation error $x_t - \xi_t$. In contrast, the accumulated state correlation matrix $\Omega_{K,L}$ lacks such good properties, rendering it challenging to derive that the derivative of the cost function about the controller gain vanishes, as shown in \Cref{proposition.stationary_failed}. 


\begin{proposition}
\label{proposition.stationary_failed}
When $K^{\star} (\Omega_{22}-\Omega_{12}^{\tr}) \neq 0$ and the observation error correlation $Y_{22}$ is not equal to the cross-correlation $Y_{12}^{\tr}$ with system state, $K^{\star}$ defined in~\eqref{eq.standard_K} and $L^{\star}$ defined in~\eqref{eq.standard_L} cannot constitute the stationary point of \Cref{problem.OD-LQR}. 
\end{proposition}
\begin{proof}
We will show that if $\nabla_{K}J(K,L^{\star})|_{K=K^{\star}} = 0 $, then $K^{\star} (\Omega_{22}-\Omega_{12}^{\tr}) = 0$ or $Y_{22} - Y_{12}^{\tr} = 0$. Similar to~\eqref{eq.gradient_L1}, the policy gradient $\nabla_K J(K,L)$ in~\eqref{eq.gradient} can be expanded as
\begin{equation}
\label{eq.gradient_K1}
\begin{aligned}
\nabla_{K}J(K,L) &= 2 (R + B^{\tr} S_{11} B) K \Sigma_{22} \\
&\quad \ - 2B^{\tr} S_{11} A (\Omega_{11} - \Omega_{12}) \\
&\quad \ - 2B^{\tr} S_{12} (A - L C) (\Omega_{12}^{\tr} - \Omega_{22}), 
\end{aligned}
\end{equation}
where $\Sigma_{22} := \mathbb{E}_{\bar{z}_0\sim \mathcal{B}} [\sum_{t=0}^{\infty}\xi_t \xi_t^{\tr}] = \Omega_{11} - \Omega_{12} -\Omega_{12}^{\tr} + \Omega_{22}$ denotes the accumulated correlation of the internal state. 

Note that we have $K^{\star} = (R + B^{\tr} {\hat{S}^{\star}} B)^{-1} B^{\tr} {\hat{S}^{\star}} A$ and $S_{11}^{\star} = \hat{S}^{\star}$. Substituting~\eqref{eq.standard_K} and~\eqref{eq.standard_L} into~\eqref{eq.gradient_K1} yields 
\begin{equation}
\nonumber
\begin{aligned}
\nabla_{K}J(K,L^{\star})|_{K=K^{\star}} = 2 (R + B^{\tr} {\hat{S}^{\star}} B) K^{\star} (\Omega_{22}-\Omega_{12}^{\tr}).
\end{aligned}
\end{equation}
Subtracting the transpose of~\eqref{eq.block_lyapunov_Sigma12} from~\eqref{eq.block_lyapunov_Sigma22}, and applying the standard LQR controller~\eqref{eq.standard_K} and observer~\eqref{eq.standard_L}, one has
\begin{equation}
\resizebox{1.0\hsize}{!}{$
\begin{aligned}
\Omega_{22}-\Omega_{12}^{\tr} &= Y_{22} - Y_{12}^{\tr} + \left(A - L^{\star} C\right) \Omega_{22} \left(A - B K^{\star} - L^{\star} C\right)^{\tr} \\ 
&\quad \ - \left(A - L^{\star} C\right) \Omega_{12}^{\tr} \left(A - B K^{\star}\right)^{\tr} \\
&= Y_{22} - Y_{12}^{\tr} + (A - L^{\star} C) (\Omega_{22} - \Omega_{12}^{\tr}) (A - B K^{\star})^{\tr}. 
\end{aligned}
$}
\end{equation}
Since $R + B^{\tr} {\hat{S}^{\star}} B \succ 0$, if $\nabla_{K}J(K,L^{\star})|_{K=K^{\star}} = 0$, then $K^{\star} (\Omega_{22}-\Omega_{12}^{\tr}) = 0$ or $Y_{22} - Y_{12}^{\tr} = 0$, which completes the proof. It is also clear that $K^{\star}$ is generally suboptimal in terms of achieving the minimal cost under the standard observer $L^{\star}$ for general initial state correlation matrix $Y$. 
\end{proof}

The above proposition excludes the possibility of the standard LQR controller $K^{\star}$ and observer $L^{\star}$ forming the stationary point of \Cref{problem.OD-LQR} unless $K^{\star} (\Omega_{22}-\Omega_{12}^{\tr}) = 0$ or 
\begin{equation}
\nonumber
\begin{aligned}
Y_{22} - Y_{12}^{\tr} &= \mathbb{E}_{\bar{z}_0 \sim \mathcal{B}} \left[(x_0 - \xi_0)(x_0 - \xi_0)^{\tr} - (x_0 - \xi_0)x_0^{\tr}\right] \\
&= - \mathbb{E}_{\bar{z}_0 \sim \mathcal{B}} [(x_0 - \xi_0)\xi_0^{\tr}] = 0, 
\end{aligned}
\end{equation}
that is, the initial cross-correlation between observation error $x_0 - \xi_0$ and state observation $\xi_0$ exactly vanishes. 
Subsequent numerical experiments in Section~\ref{sec.experiment} will present that the stationary point obtained through numerical methods are generally independent of the standard LQR controller $K^{\star}$ and the standard observer $L^{\star}$. The following theorem will combine the policy gradient expressions in \Cref{lemma.gradient} and derive the conditions that the stationary point satisfies. 

\begin{theorem}
\label{theorem.sufficient_condition}
Suppose that $B$ has columns of full rank. The stationary point of $J(K^{\ddagger}, L^{\ddagger})$ in \Cref{problem.OD-LQR} satisfies the following discrete-time Sylvester equations: 
\begin{subequations}
\begin{align}
    \label{eq.sylvester_equation_k}
    K^{\ddagger} &= G_{K} + R_{K}^{-1} M K^{\ddagger} N \Sigma_{22}^{-1}, \\
    \label{eq.sylvester_equation_l}
    L^{\ddagger} &= G_{L} + S_{22}^{-1} U L^{\ddagger} V (C \Omega_{22} C^{\tr})^{-1}, 
\end{align}
\end{subequations}
where 
\begin{equation}
\nonumber
\begin{aligned}
    R_{K} &= R + B^{\tr} S_{11} B, \\
    M &= B^{\tr} S_{12} S_{22}^{-1} S_{12}^{\tr} B, \\
    N &= (\Omega_{12} - \Omega_{22}) C^{\tr} (C \Omega_{22} C^{\tr})^{-1} C (\Omega_{12}^{\tr} - \Omega_{22}), \\
    G_{K} &= K^{\circ} - R_{K}^{-1} B^{\tr} S_{12} L^{\circ}C (\Omega_{12}^{\tr} - \Omega_{22}) \Sigma_{22}^{-1}, \\
    U &= S_{12}^{\tr} B R_{K}^{-1} B^{\tr} S_{12}, \\
    V &= C (\Omega_{12}^{\tr} - \Omega_{22}) \Sigma_{22}^{-1} (\Omega_{12} - \Omega_{22}) C^{\tr}, \\
    G_{L} &= L^{\circ} - S_{22}^{-1} S_{12}^{\tr} BK^{\circ} (\Omega_{12} - \Omega_{22}) C^{\tr} (C \Omega_{22} C^{\tr})^{-1}, \\
    K^{\circ} &= R_{K}^{-1} B^{\tr} S_{11} A + R_{K}^{-1} B^{\tr} (S_{11} + S_{12}) A (\Omega_{12}^{\tr} - \Omega_{22}) \Sigma_{22}^{-1}, \\
    L^{\circ} &= A \Omega_{22} C^{\tr} (C \Omega_{22} C^{\tr})^{-1} + S_{22}^{-1} S_{12}^{\tr} A \Omega_{12} C^{\tr} (C \Omega_{22} C^{\tr})^{-1}. 
\end{aligned}
\end{equation}
\end{theorem}
\begin{proof}
Letting the gradients~\eqref{eq.gradient_L1} and~\eqref{eq.gradient_K1} vanish, we have 
\begin{subequations}
\begin{align}
    \nonumber
    K^{\ddagger} &= R_{K}^{-1} B^{\tr} S_{11} A + R_{K}^{-1} B^{\tr} (S_{11} + S_{12}) A (\Omega_{12}^{\tr} - \Omega_{22}) \Sigma_{22}^{-1} \\
    \nonumber
    &\quad \ - R_{K}^{-1} B^{\tr} S_{12} L^{\ddagger} C (\Omega_{12}^{\tr} - \Omega_{22}) \Sigma_{22}^{-1} \\
    &= K^{\circ} - R_{K}^{-1} B^{\tr} S_{12} L^{\ddagger} C (\Omega_{12}^{\tr} - \Omega_{22}) \Sigma_{22}^{-1}, \\
    \nonumber
    L^{\ddagger} &= A \Omega_{22} C^{\tr} (C \Omega_{22} C^{\tr})^{-1} + S_{22}^{-1} S_{12}^{\tr} A \Omega_{12} C^{\tr} (C \Omega_{22} C^{\tr})^{-1} \\
    \nonumber
    &\quad \ - S_{22}^{-1} S_{12}^{\tr} B K^{\ddagger} (\Omega_{12} - \Omega_{22}) C^{\tr} (C \Omega_{22} C^{\tr})^{-1} \\
    &= L^{\circ} \!-\! S_{22}^{-1} S_{12}^{\tr} B K^{\ddagger} (\Omega_{12} \!-\! \Omega_{22}) C^{\tr} (C \Omega_{22} C^{\tr})^{-1}. 
\end{align}
\end{subequations}
Combining the above equations yields 
\begin{equation}
\nonumber
\begin{aligned}
    K^{\ddagger} &= K^{\circ} - R_{K}^{-1} B^{\tr} S_{12} L^{\circ}C (\Omega_{12}^{\tr} - \Omega_{22}) \Sigma_{22}^{-1} \\
    &\quad \ + R_{K}^{-1} M K^{\ddagger} N \Sigma_{22}^{-1}, \\
    L^{\ddagger} &= L^{\circ} - S_{22}^{-1} S_{12}^{\tr} BK^{\circ} (\Omega_{12} - \Omega_{22}) C^{\tr} (C \Omega_{22} C^{\tr})^{-1} \\
    &\quad \ + S_{22}^{-1} U L^{\ddagger} V (C \Omega_{22} C^{\tr})^{-1}, 
\end{aligned}
\end{equation}
which completes the proof. 
\end{proof}

The derived discrete-time Sylvester equations can be solved by utilizing the \texttt{dlyap} instruction in MATLAB. 
The Sylvester equation~\eqref{eq.sylvester_equation_k} about the controller $K^{\ddagger}$ has a unique solution if and only if $\lambda_{i}\mu_{j} \neq 1$ for all $i = 1, \cdots, m$ and $j = 1, \cdots, n$, where $\lambda_{1}, \cdots, \lambda_{m}$ are the eigenvalues of $R_{K}^{-1} M$, and $\mu_{1}, \cdots, \mu_{n}$ are the eigenvalues of $N \Sigma_{22}^{-1}$, and the same applies to the Sylvester equation~\eqref{eq.sylvester_equation_l} about the observer $L^{\ddagger}$. 

Unlike the independent assignment of the poles of feedback controller and state observer, the controller $K^{\ddagger}$ and observer $L^{\ddagger}$ depend on each other and jointly affect the objective function. Although the derived Sylvester equations are coupled, they have good symmetry. 
By utilizing the influence of the standard controller-observer pair $(K^{\star}, L^{\star})$ on matrices $S_{12}$ and $\Omega_{12}^{\tr} - \Omega_{22}$, the following Proposition will further investigate the relationship between the stationary point $(K^{\ddagger}, L^{\ddagger})$ and the standard controller-observer pair $(K^{\star}, L^{\star})$. 

\begin{proposition}
\label{proposition.stationary_degeneration}
Given the standard observer $L^{\star}$ or the standard LQR controller $K^{\star}$ with the special initial state correlation satisfying $Y_{22} - Y_{12}^{\tr} = 0$, the derived stationary point $(K^{\ddagger}, L^{\ddagger})$ will degrade into the ordinary pair $(K^{\star}, L^{\star})$. 
\end{proposition}
\begin{proof}
Considering the special initial correlation $Y_{22} - Y_{12}^{\tr} = 0$, one has $\Omega_{22} - \Omega_{12}^{\tr} = 0$. Therefore, $N = 0$, $G_{K} = K^{\circ}$, $V = 0$, $G_{L} = L^{\circ}$, and the Sylvester equations can be simplified to 
\begin{equation}
\nonumber
\begin{aligned}
    K^{\ddagger} &= \left(R + B^{\tr} S_{11} B\right)^{-1} B^{\tr} S_{11} A, \\
    L^{\ddagger} &= A \Omega_{22} C^{\tr} (C \Omega_{22} C^{\tr})^{-1} + S_{22}^{-1} S_{12}^{\tr} A \Omega_{12} C^{\tr} (C \Omega_{22} C^{\tr})^{-1}. 
\end{aligned}
\end{equation}
Since $S_{11}$ only depends on $K^{\ddagger}$, $K^{\ddagger} = K^{\star}$. According to \Cref{lemma.positive_P}, $S_{12} = 0$, and $L^{\ddagger} = A \Omega_{22} C^{\tr} (C \Omega_{22} C^{\tr})^{-1}$. Since $\Omega_{22}$ only depends on $L^{\ddagger}$, we find that $L^{\ddagger} = L^{\star}$ holds exactly. 

Given the standard LQR controller $K^{\star}$, according to \Cref{lemma.positive_P}, $S_{12} = 0$. Thus, $M = 0$, $G_{K} = K^{\circ}$, $U = 0$, $G_{L} = L^{\circ}$, and the Sylvester equations can be simplified to 
\begin{equation}
\nonumber
\begin{aligned}
    K^{\ddagger} &= R_{K}^{-1} B^{\tr} S_{11} A + R_{K}^{-1} B^{\tr} (S_{11} + S_{12}) A (\Omega_{12}^{\tr} - \Omega_{22}) \Sigma_{22}^{-1}, \\
    L^{\ddagger} &= A \Omega_{22} C^{\tr} (C \Omega_{22} C^{\tr})^{-1}. 
\end{aligned}
\end{equation}
Since $\Omega_{22}$ only depends on $L^{\ddagger}$, $L^{\ddagger} = L^{\star}$. Similarly, considering the special initial correlation $Y_{22} - Y_{12}^{\tr} = 0$, one has $\Omega_{22} - \Omega_{12}^{\tr} = 0$ and $K^{\ddagger} = R_{K}^{-1} B^{\tr} S_{11} A$. Since $S_{11}$ only depends on $K^{\ddagger}$, we find that $K^{\ddagger} = K^{\star}$ holds exactly. 
\end{proof}

\begin{remark}
The condition $Y_{22} = Y_{12}^{\tr}$ implies $\Omega_{22} = \Omega_{12}^{\tr}$ for the derived stationary point, which reduces to the standard pair and can be designed separately. This statistical orthogonality eliminates the influence of initial correlation on the optimal gains and recovers the certainty equivalence property. 
\end{remark}

\section{Experiment}
\label{sec.experiment}
In this section, we will compare the costs and gradients of the standard LQR controller-observer pair and the stationary point of \texttt{OD-LQR} obtained through numerical methods. 

\subsection{Internally Unstable Linear System}
\label{sec.unstable_example}
Take the discrete version of the Doyle's LQG example~\cite{doyle1978guaranteed} 
\begin{equation}
\label{eq.unstable_system}
    A = \begin{bmatrix}
    1.1 & 0.1 \\
    0 & 1.1
    \end{bmatrix}, 
    B = \begin{bmatrix}
    0 \\
    0.1
    \end{bmatrix}, 
    C = \begin{bmatrix}
    1.0 & 1.0
    \end{bmatrix},
\end{equation}
which is derived from a physical double-integrator model under a specific coordinate transformation. The selected plant is internally unstable. Experiment results for stable systems are provided in the open-source code.\footnote{The code is available at \url{https://github.com/jieli18/od_lqr}} Let $Q = 0.25 I_2$, $R = 0.2$. The standard LQR controller $K^{\star} = \begin{bmatrix} 4.8768 & 4.3773 \end{bmatrix}$ can be directly obtained by solving the discrete-time algebraic Riccati equation~\eqref{eq.P_riccati}. Rather, the standard observer minimizing the accumulated estimation variance is determined by the initial estimation variance $E_0$. The following will explore the results of general and special initial state correlations separately: 
\begin{equation}
\nonumber
Y_g = \begin{bmatrix}
    2 & 0 & 0.1 & 0 \\
    0 & 2 & 0 & 0.1 \\
    0.1 & 0 & 1 & 0 \\
    0 & 0.1 & 0 & 1 \\
\end{bmatrix}, 
Y_s = \begin{bmatrix}
    2 & 0 & 1 & 0 \\
    0 & 2 & 0 & 1 \\
    1 & 0 & 1 & 0 \\
    0 & 1 & 0 & 1 \\
\end{bmatrix}. 
\end{equation}

\textit{1) General Initial State Correlation $Y_g$:} 
The standard observer $L^{\star} = \begin{bmatrix} -0.5667 & 1.8333 \end{bmatrix}^{\tr}$ can be solved through the algebraic Riccati equation~\eqref{eq.sigma_riccati}. The cost function of $L^{\star}$ under the standard LQR controller $K^{\star}$ is shown by the red dot in Fig.~\ref{subfig:perf_1_obs}, and the white space denotes unstable $\hat{\mathcal{A}}_{K,L}$, where the Lyapunov equations are unsolvable and the cost value is undefined. The unique minimum obtained by numerical search indicates that the standard observer $L^{\star}$ achieves optimality given the standard LQR controller $K^{\star}$. The norm of the gradient of cost function with respect to $L$ depicted in Fig.~\ref{subfig:grad_1_obs} shows that $L^{\star}$ is a stationary point of $J(K^{\star}, L)$. 

The landscape of the cost function of controller $K$ under the standard observer $L^{\star}$ is shown in Fig.~\ref{subfig:perf_1_ctr}, where the cost value of the minimum point is smaller than that of the standard LQR controller $K^{\star}$. Besides, the gradient of cost function for the minimum point vanishes, as shown in Fig.~\ref{subfig:grad_1_ctr}. By employing the \texttt{dlyap} instruction, the stationary point $K^{\ddagger} = \begin{bmatrix} 4.2598 & 3.9482 \end{bmatrix}$ and $L^{\ddagger} = \begin{bmatrix} -2.5604 & 4.0196 \end{bmatrix}^{\tr}$ derived in \Cref{theorem.sufficient_condition}, whose cost value is 102.2875, can be obtained numerically. Therefore, the experimental results indicate that the standard controller-observer pair $(K^{\star}, L^{\star})$ is generally different from the stationary point of the \texttt{OD-LQR} problem, and the standard LQR controller $K^{\star}$ is generally suboptimal when paired with the standard observer $L^{\star}$. 

\begin{figure}[!tb]
    \centering

    \subfigure[Cost function under $K^{\star}$]{
        \label{subfig:perf_1_obs}
        \begin{minipage}[t]{0.45\linewidth}
        \includegraphics[width=0.99\textwidth,trim=0 0 0 12,clip]{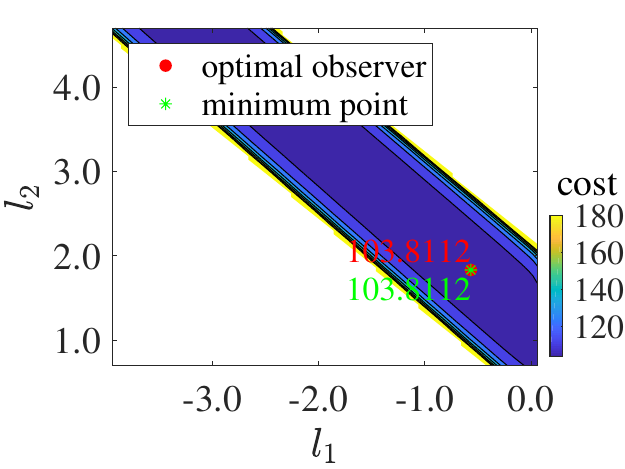}  
        \end{minipage}
    }
    \subfigure[Norm of grad w.r.t $L$ under $K^{\star}$]{
        \label{subfig:grad_1_obs}
        \begin{minipage}[t]{0.45\linewidth}
        \includegraphics[width=0.99\textwidth,trim=0 0 0 12,clip]{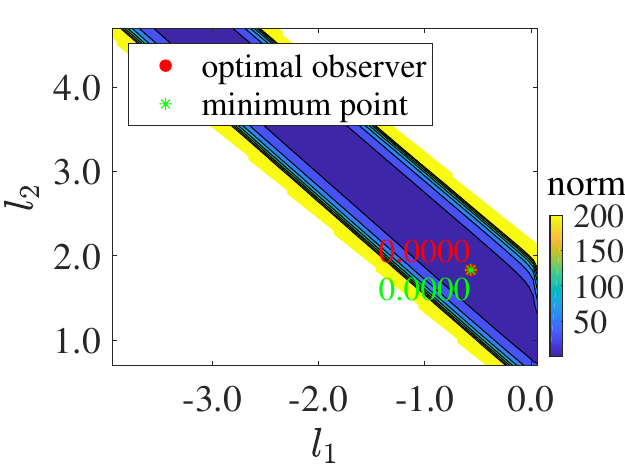}
        \end{minipage}
    }
    \subfigure[Cost function under $L^{\star}$]{
        \label{subfig:perf_1_ctr}
        \begin{minipage}[t]{0.45\linewidth}
        \includegraphics[width=0.99\textwidth,trim=0 0 0 12,clip]{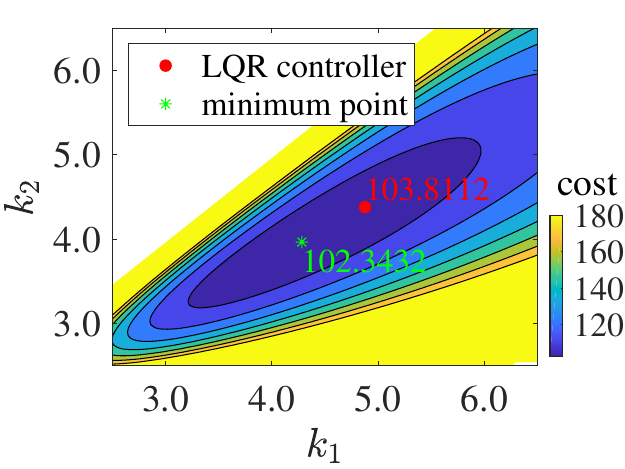}
        \end{minipage}
    }
    \subfigure[Norm of grad w.r.t $K$ under $L^{\star}$]{
        \label{subfig:grad_1_ctr}
        \begin{minipage}[t]{0.45\linewidth}
        \includegraphics[width=0.99\textwidth,trim=0 0 0 12,clip]{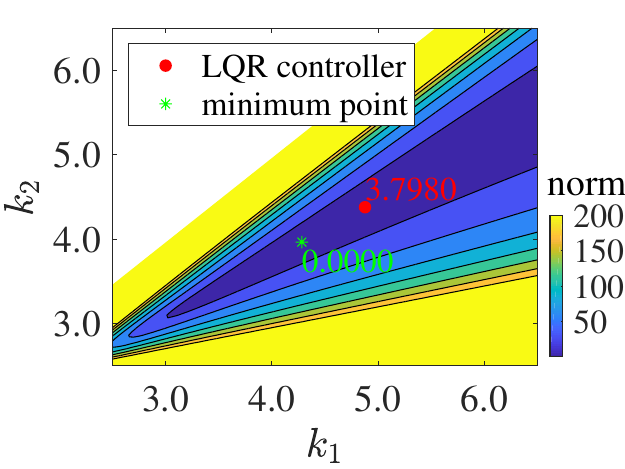}
        \end{minipage}
    }
    \vspace{-2mm}
    \caption{Results of general initial state correlation}
    \label{fig:perf_1}
    \vspace{-1mm}
\end{figure}

\textit{2) Special Initial State Correlation $Y_s$:} 
Since $Y_{22}$ is the same as the previous case, the standard observer $L^{\star}$ is also the same. The cost values of various observers under $K^{\star}$ are depicted in Fig.~\ref{subfig:perf_2_obs}. It can be found that the standard observer $L^{\star}$ achieves the minimum cost value when paired with the standard LQR controller $K^{\star}$. Similarly, the gradient norm of the cost function for diverse observers shown in Fig.~\ref{subfig:grad_2_obs} indicates that $L^{\star}$ is a stationary point of $J(K^{\star}, L)$. 

The cost function of different controllers operating with the standard observer $L^{\star}$ is illustrated in Fig.~\ref{subfig:perf_2_ctr}. Different from the previous case, the minimal cost value identified through numerical grid search coincides exactly with the cost value of the standard LQR controller $K^{\star}$. Besides, as shown in Fig.~\ref{subfig:grad_2_ctr}, the gradient of cost function for the standard LQR controller $K^{\star}$ vanishes. Therefore, for the special initial correlation $Y_{22} - Y_{12}^{\tr} = 0$, the standard controller-observer pair $(K^{\star}, L^{\star})$ is a stationary point, which shows control optimality. 

\begin{figure}[!tb]
    \centering

    \subfigure[Cost function under $K^{\star}$]{
        \label{subfig:perf_2_obs}
        \begin{minipage}[t]{0.45\linewidth}
        \includegraphics[width=0.99\textwidth,trim=0 0 0 12,clip]{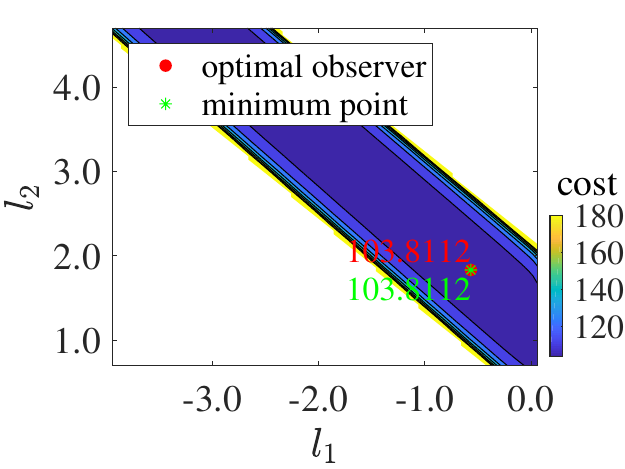}
        \end{minipage}
    }
    \subfigure[Norm of grad w.r.t $L$ under $K^{\star}$]{
        \label{subfig:grad_2_obs}
        \begin{minipage}[t]{0.45\linewidth}
        \includegraphics[width=0.99\textwidth,trim=0 0 0 12,clip]{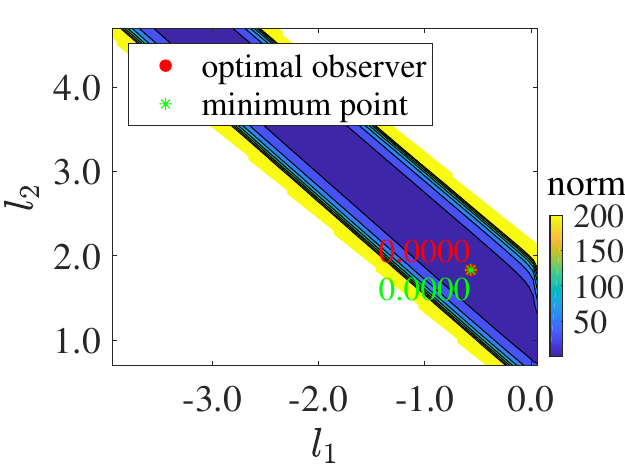}
        \end{minipage}
    }
    \qquad
    \subfigure[Cost function under $L^{\star}$]{
        \label{subfig:perf_2_ctr}
        \begin{minipage}[t]{0.45\linewidth}
        \includegraphics[width=0.99\textwidth,trim=0 0 0 12,clip]{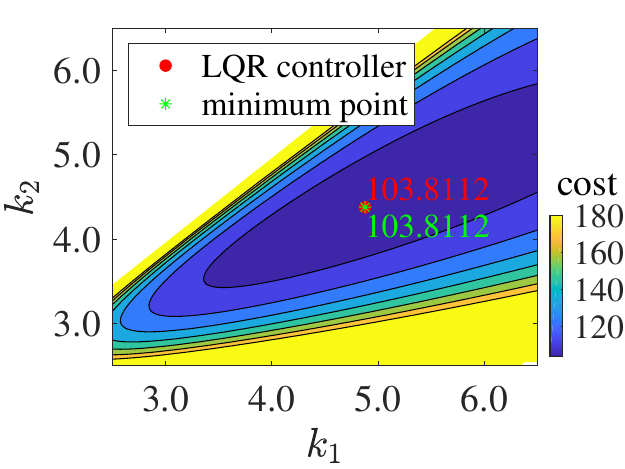}
        \end{minipage}
    }
    \subfigure[Norm of grad w.r.t $K$ under $L^{\star}$]{
        \label{subfig:grad_2_ctr}
        \begin{minipage}[t]{0.45\linewidth}
        \includegraphics[width=0.99\textwidth,trim=0 0 0 12,clip]{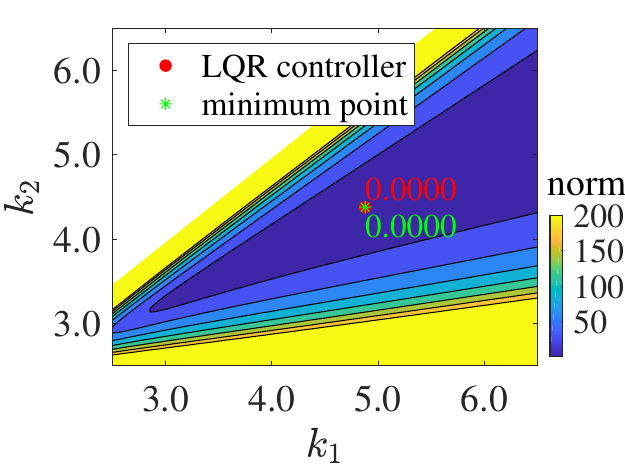}
        \end{minipage}
    }
    \vspace{-2mm}
    \caption{Results of special initial state correlation}
    \label{fig:perf_2}
    \vspace{-1mm}
\end{figure}

\subsection{Unstable System with 2D Controller and 2D Observer}
Consider the system with complex controller and observer: 
\begin{equation}
\nonumber
    A = \begin{bmatrix}
    1.1 & 0.1 \\
    0   & 1.1
    \end{bmatrix}, 
    B = \begin{bmatrix}
    0   & 0.1 \\
    0.1 & 0
    \end{bmatrix}, 
    C = \begin{bmatrix}
    1.0 & 1.0 \\
    0   & 1.0
    \end{bmatrix}.
\end{equation}
For $Y_g$ and $Y_s$, we separately solved the standard controller-observer pair $(K^{\star}, L^{\star})$ and identified the stationary point $(K^{\ddagger}, L^{\ddagger})$ of \texttt{OD-LQR}. 
The results in Table~\ref{tab:comp} indicate that, for the general initial state correlation $Y_g$, the stationary point $(K^{\ddagger}, L^{\ddagger})$ is different from $(K^{\star}, L^{\star})$, and achieves a slightly lower cost value compared with the standard LQR; while for the special initial state correlation $Y_s$ with $Y_{22} = Y_{12}^{\tr}$, the cost values are consistent. The results confirm that the theory remain valid for the higher-dimensional system. 

\begin{table}[!ht]
\centering
\caption{Comparison of cost values}
\label{tab:comp}
\begin{tabular}{ccc}
    \toprule
    $Y$ & $J(K^{\star}, L^{\star})$ & $J(K^{\ddagger}, L^{\ddagger})$ \\
    \midrule
    $Y_g$ & 25.4400 & 25.1660 \\
    $Y_s$ & 25.4400 & 25.4400 \\
    \bottomrule
\end{tabular}
\end{table}

\section{Conclusion}
\label{sec.conclusion}
In this work, we have explored the optimality of observer-controller pair on the cost function for \texttt{OD-LQR} problems. Based on the derived gradient expressions, we have demonstrated that the standard observer gain minimizing the accumulated estimation variance ensures optimality under the standard LQR controller. However, the standard LQR controller usually fails to achieve the optimal performance under the standard observer, unless the initial state correlations have a special structure. Moreover, we have characterized the stationary point of \texttt{OD-LQR} by Sylvester equations and proved that it reduces to the standard pair under the special initial state correlation, thereby recovering the separation principle. Our examples provide empirical support for the proposed theoretical results. Beyond providing practical guidance for separation-based controller designs, the derived Sylvester equations open a distinct pathway for data-driven methods in dynamic control, shifting the paradigm away from traditional Bellman equations.

\ifCLASSOPTIONcaptionsoff
  \newpage
\fi
\bibliographystyle{ieeetr}
\bibliography{reference}

\newpage
\onecolumn

\appendix
\section*{Proof of \Cref{lemma.gradient_dominance}}
\label{sec.gradient_dominance}
\begin{proof}
We define the cost difference between arbitrary gains $(K,L)$ and the optimal gains $(K^{\ddagger}, L^{\ddagger})$ as 
$$
J(K,L) - J(K^{\ddagger}, L^{\ddagger}) = \text{Tr}\left( (S_{K,L} - S^{\ddagger}) Y \right) = \text{Tr}\left( \Delta S \cdot Y \right). 
$$

First, we establish the relationship between $\Delta S$ and the parameter errors $\Delta K = K - K^{\ddagger}$ and $\Delta L = L - L^{\ddagger}$, and derive the cost difference through the Lyapunov equation. 

The matrices $S_{K,L}$ and $S^{\ddagger}$ satisfy the following Lyapunov equations: 
$$
\begin{aligned}
S_{K,L} &= \hat{Q}_K + \hat{\mathcal{A}}_{K,L}^\top S_{K,L} \hat{\mathcal{A}}_{K,L}, \\ 
S^{\ddagger} &= \hat{Q}_{K^{\ddagger}} + (\hat{\mathcal{A}}^{\ddagger})^\top S^{\ddagger} \hat{\mathcal{A}}^{\ddagger}, 
\end{aligned}
$$
where $\hat{Q}_K := \begin{bmatrix}Q + K^\top R K & - K^\top R K\\- K^\top R K & K^\top R K\end{bmatrix}$, and $\hat{\mathcal{A}}^{\ddagger} := \hat{\mathcal{A}}_{K^{\ddagger}, L^{\ddagger}} = \bar{A} - \hat{B} K^{\ddagger} \bar{F} + \hat{F}^\top L^{\ddagger} \hat{C}$. Subtracting these two equations yields 
$$
\Delta S = S_{K,L} - S^{\ddagger} = \hat{Q}_K - \hat{Q}_{K^{\ddagger}} + \hat{\mathcal{A}}_{K,L}^\top S_{K,L} \hat{\mathcal{A}}_{K,L} - (\hat{\mathcal{A}}^{\ddagger})^\top S^{\ddagger} \hat{\mathcal{A}}^{\ddagger}.
$$
Let $\Delta \hat{\mathcal{A}} = \hat{\mathcal{A}}_{K,L} - \hat{\mathcal{A}}^{\ddagger}$. Substituting $\hat{\mathcal{A}}_{K,L} = \hat{\mathcal{A}}^{\ddagger} + \Delta \hat{\mathcal{A}}$ and $S_{K,L} = S^{\ddagger} + \Delta S$ into the above equation yields 
$$
\begin{aligned}
\Delta S =& \hat{Q}_K - \hat{Q}_{K^{\ddagger}} + (\hat{\mathcal{A}}^{\ddagger} + \Delta \hat{\mathcal{A}})^\top (S^{\ddagger} + \Delta S) (\hat{\mathcal{A}}^{\ddagger} + \Delta \hat{\mathcal{A}}) - (\hat{\mathcal{A}}^{\ddagger})^\top S^{\ddagger} \hat{\mathcal{A}}^{\ddagger} \\
=& \hat{Q}_K - \hat{Q}_{K^{\ddagger}} + (\hat{\mathcal{A}}^{\ddagger})^\top S^{\ddagger} \Delta \hat{\mathcal{A}} + \Delta \hat{\mathcal{A}}^\top S^{\ddagger} \hat{\mathcal{A}}^{\ddagger} + \Delta \hat{\mathcal{A}}^\top S^{\ddagger} \Delta \hat{\mathcal{A}} + \hat{\mathcal{A}}_{K,L}^\top \Delta S \hat{\mathcal{A}}_{K,L}.
\end{aligned}
$$
By rearranging terms, we obtain a Lyapunov equation for $\Delta S$ 
$$
\Delta S = \mathcal{Z} + \hat{\mathcal{A}}_{K,L}^\top \Delta S \hat{\mathcal{A}}_{K,L},
$$
where the remainder term $\mathcal{Z}$ is defined as 
$$
\mathcal{Z} := \hat{Q}_K - \hat{Q}_{K^{\ddagger}} + (\hat{\mathcal{A}}^{\ddagger})^\top S^{\ddagger} \Delta \hat{\mathcal{A}} + \Delta \hat{\mathcal{A}}^\top S^{\ddagger} \hat{\mathcal{A}}^{\ddagger} + \Delta \hat{\mathcal{A}}^\top S^{\ddagger} \Delta \hat{\mathcal{A}}.
$$
Since $\hat{\mathcal{A}}_{K,L}$ is stable, the solution for $\Delta S$ is 
$$
\Delta S = \sum_{t=0}^\infty (\hat{\mathcal{A}}_{K,L}^\top)^t \mathcal{Z} (\hat{\mathcal{A}}_{K,L})^t. 
$$
Consequently, the cost difference is 
$$
J(K,L) - J(K^{\ddagger}, L^{\ddagger}) = \text{Tr}(\Delta S Y) = \text{Tr}\left( \mathcal{Z} \sum_{t=0}^\infty (\hat{\mathcal{A}}_{K,L}^\top)^t Y (\hat{\mathcal{A}}_{K,L})^t \right) = \text{Tr}(\mathcal{Z} \Omega_{K,L}). 
$$

Next, we expand the remainder term $\mathcal{Z}$ in the Lyapunov equation for $\Delta S$. 

Let $\Delta K = K - K^{\ddagger}$ and $\Delta L = L - L^{\ddagger}$. The dynamics variation is 
$$
\Delta \hat{\mathcal{A}} = \hat{\mathcal{A}}_{K,L} - \hat{\mathcal{A}}^{\ddagger} = -\hat{B} \Delta K \bar{F} + \hat{F}^\top \Delta L \hat{C}. 
$$
The cost matrix variation is 
$$
\hat{Q}_K - \hat{Q}_{K^{\ddagger}} = \bar{F}^\top (K^\top R K - (K^{\ddagger})^\top R K^{\ddagger}) \bar{F}. 
$$
Using the identity $K^\top R K - (K^{\ddagger})^\top R K^{\ddagger} = \Delta K^\top R K^{\ddagger} + (K^{\ddagger})^\top R \Delta K + \Delta K^\top R \Delta K$, we write 
$$
\hat{Q}_K - \hat{Q}_{K^{\ddagger}} = \bar{F}^\top \Delta K^\top R K^{\ddagger} \bar{F} + \bar{F}^\top (K^{\ddagger})^\top R \Delta K \bar{F} + \bar{F}^\top \Delta K^\top R \Delta K \bar{F}. 
$$
Substituting these into the expression for the remainder term $\mathcal{Z}$ in the Lyapunov equation for $\Delta S$ yields 
$$
\begin{aligned}
\mathcal{Z} &= \bar{F}^\top \Delta K^\top R K^{\ddagger} \bar{F} + \bar{F}^\top (K^{\ddagger})^\top R \Delta K \bar{F} + \bar{F}^\top \Delta K^\top R \Delta K \bar{F} \\
&\quad - (\hat{\mathcal{A}}^{\ddagger})^\top S^{\ddagger} \hat{B} \Delta K \bar{F} + (\hat{\mathcal{A}}^{\ddagger})^\top S^{\ddagger} \hat{F}^\top \Delta L \hat{C} - \bar{F}^\top \Delta K^\top \hat{B}^\top S^{\ddagger} \hat{\mathcal{A}}^{\ddagger} + \hat{C}^\top \Delta L^\top \hat{F} S^{\ddagger} \hat{\mathcal{A}}^{\ddagger} \\
&\quad + (-\bar{F}^\top \Delta K^\top \hat{B}^\top + \hat{C}^\top \Delta L^\top \hat{F}) S^{\ddagger} (-\hat{B} \Delta K \bar{F} + \hat{F}^\top \Delta L \hat{C}) \\
&= \bar{F}^\top \Delta K^\top R K^{\ddagger} \bar{F} + \bar{F}^\top (K^{\ddagger})^\top R \Delta K \bar{F} - (\hat{\mathcal{A}}^{\ddagger})^\top S^{\ddagger} \hat{B} \Delta K \bar{F} - \bar{F}^\top \Delta K^\top \hat{B}^\top S^{\ddagger} \hat{\mathcal{A}}^{\ddagger} \\
&\quad + (\hat{\mathcal{A}}^{\ddagger})^\top S^{\ddagger} \hat{F}^\top \Delta L \hat{C} + \hat{C}^\top \Delta L^\top \hat{F} S^{\ddagger} \hat{\mathcal{A}}^{\ddagger} \\
&\quad + \bar{F}^\top \Delta K^\top R \Delta K \bar{F} + (-\bar{F}^\top \Delta K^\top \hat{B}^\top + \hat{C}^\top \Delta L^\top \hat{F}) S^{\ddagger} (-\hat{B} \Delta K \bar{F} + \hat{F}^\top \Delta L \hat{C}). 
\end{aligned}
$$

We group the linear terms about $\Delta K$ in $\mathcal{L}_K^{\ddagger}$ as follows: 
$$
\begin{aligned}
\mathcal{L}_K^{\ddagger} &= \bar{F}^\top \Delta K^\top R K^{\ddagger} \bar{F} + \bar{F}^\top (K^{\ddagger})^\top R \Delta K \bar{F} - (\hat{\mathcal{A}}^{\ddagger})^\top S^{\ddagger} \hat{B} \Delta K \bar{F} - \bar{F}^\top \Delta K^\top \hat{B}^\top S^{\ddagger} \hat{\mathcal{A}}^{\ddagger} \\
&= \bar{F}^\top \Delta K^\top \left( R K^{\ddagger} \bar{F} - \hat{B}^\top S^{\ddagger} \hat{\mathcal{A}}^{\ddagger} \right) + \left( \bar{F}^\top (K^{\ddagger})^\top R - (\hat{\mathcal{A}}^{\ddagger})^\top S^{\ddagger} \hat{B} \right) \Delta K \bar{F} \\
&= \bar{F}^\top \Delta K^\top E_K^{\ddagger} + (E_K^{\ddagger})^\top \Delta K \bar{F},
\end{aligned}
$$
where $E_K := R K \bar{F} - \hat{B}^\top S \hat{\mathcal{A}}_{K,L}$ and $E_K^{\ddagger} := R K^{\ddagger} \bar{F} - \hat{B}^\top S^{\ddagger} \hat{\mathcal{A}}^{\ddagger}$. Define the error of $E_K$ as 
$$
\begin{aligned}
\Delta E_K &= E_K - E_K^{\ddagger} \\
&= R \Delta K \bar{F} - \hat{B}^\top (S \hat{\mathcal{A}}_{K,L} - S^{\ddagger} A^{\ddagger}) \\
&= R \Delta K \bar{F} - \hat{B}^\top (S \hat{\mathcal{A}}_{K,L} - S^{\ddagger} A^{\ddagger} + S^{\ddagger} \hat{\mathcal{A}}_{K,L} - S^{\ddagger} \hat{\mathcal{A}}_{K,L}) \\
&= R \Delta K \bar{F} - \hat{B}^\top S^{\ddagger} \Delta \hat{A} - \hat{B}^\top \Delta S \hat{\mathcal{A}}_{K,L} \\
&= R \Delta K \bar{F} + \hat{B}^\top S^{\ddagger} \hat{B} \Delta K \bar{F} - \hat{B}^\top S^{\ddagger} \hat{F}^\top \Delta L \hat{C} - \hat{B}^\top \Delta S \hat{\mathcal{A}}_{K,L} \\
&= (R + \hat{B}^\top S^{\ddagger} \hat{B}) \Delta K \bar{F} - \hat{B}^\top S^{\ddagger} \hat{F}^\top \Delta L \hat{C} - \hat{B}^\top \Delta S \hat{\mathcal{A}}_{K,L} . 
\end{aligned}
$$

Similarly, the linear terms about $\Delta L$ is 
$$
\mathcal{L}_L^{\ddagger} = (\hat{\mathcal{A}}^{\ddagger})^\top S^{\ddagger} \hat{F}^\top \Delta L \hat{C} + \hat{C}^\top \Delta L^\top \hat{F} S^{\ddagger} \hat{\mathcal{A}}^{\ddagger} = (E_L^{\ddagger})^\top \Delta L \hat{C} + \hat{C}^\top \Delta L^\top E_L^{\ddagger}, 
$$
where $E_L = \hat{F} S \hat{\mathcal{A}}$ and $E_L^{\ddagger} = \hat{F} S^{\ddagger} \hat{\mathcal{A}}^{\ddagger}$. Define the error of $E_L$ as 
$$
\begin{aligned}
\Delta E_L &= E_L - E_L^{\ddagger} \\
&= \hat{F} (S \hat{\mathcal{A}} - S^{\ddagger} \hat{\mathcal{A}}^{\ddagger}) \\
&= \hat{F} S^{\ddagger} \Delta \hat{A} + \hat{F} \Delta S \hat{\mathcal{A}}_{K,L} \\
&= \hat{F} S^{\ddagger} \hat{F}^\top \Delta L \hat{C} - \hat{F} S^{\ddagger}\hat{B} \Delta K \bar{F} + \hat{F} \Delta S \hat{\mathcal{A}}_{K,L} . 
\end{aligned}
$$

Besides, the quadratic terms are 
$$
\begin{aligned}
\mathcal{Q}_{K,L}^{\ddagger} &= \bar{F}^\top \Delta K^\top R \Delta K \bar{F} + (-\bar{F}^\top \Delta K^\top \hat{B}^\top + \hat{C}^\top \Delta L^\top \hat{F}) S^{\ddagger} (-\hat{B} \Delta K \bar{F} + \hat{F}^\top \Delta L \hat{C}) \\
&= \bar{F}^\top \Delta K^\top (R + \hat{B}^\top S^{\ddagger} \hat{B}) \Delta K \bar{F} + \hat{C}^\top \Delta L^\top \hat{F} S^{\ddagger} \hat{F}^\top \Delta L \hat{C} \\
&\quad - \bar{F}^\top \Delta K^\top \hat{B}^\top S^{\ddagger} \hat{F}^\top \Delta L \hat{C} - \hat{C}^\top \Delta L^\top \hat{F} S^{\ddagger} \hat{B} \Delta K \bar{F} . 
\end{aligned}
$$
Since $R \succ 0$, $S^{\ddagger} \succ 0$ and $\hat{F}$ has rows of full rank, we have $R + \hat{B}^\top S^{\ddagger} \hat{B} \succ 0$ and $\hat{F} S^{\ddagger} \hat{F}^{\top} \succ 0$. Therefore, we can calculate the Frobenius norm of matrix $\mathcal{Z}$ as follows: 
$$
\begin{aligned}
\|\mathcal{Z}\|_F &\le \|\mathcal{L}_K^{\ddagger}\|_F + \|\mathcal{L}_L^{\ddagger}\|_F + \|\mathcal{Q}_{K,L}^{\ddagger}\|_F \\
&\le 2 \|\bar{F}\|_2 \|E_K^{\ddagger}\|_2 \|\Delta K\|_F + 2 \|\hat{C}\|_2 \|E_L^{\ddagger}\|_2 \|\Delta L\|_F \\
&\quad + \|\bar{F}\|_2^2 \|R + \hat{B}^\top S^{\ddagger} \hat{B}\|_2 \|\Delta K\|_F^2 + \|\hat{C}\|_2^2 \|\hat{F} S^{\ddagger} \hat{F}^\top\|_2 \|\Delta L\|_F^2 \\
&\quad + 2 \|\bar{F}\|_2 \|\hat{C}\|_2 \|\hat{B}^\top S^{\ddagger} \hat{F}^\top\|_2 \|\Delta K\|_F \|\Delta L\|_F \\
&\le 2 \|\bar{F}\|_2 \|E_K^{\ddagger} \|_2 \|\Delta K\|_F + 2 \|\hat{C}\|_2 \|E_L^{\ddagger} \|_2 \|\Delta L\|_F \\
&\quad + \|\bar{F}\|_2^2 \|R + \hat{B}^\top S^{\ddagger} \hat{B}\|_2 \|\Delta K\|_F^2 + \|\hat{C}\|_2^2 \|\hat{F} S^{\ddagger} \hat{F}^\top\|_2 \|\Delta L\|_F^2 \\
&\quad + 2 \|\bar{F}\|_2 \|\hat{C}\|_2 \|\hat{B}^\top S^{\ddagger} \hat{F}^\top\|_2 \left( \epsilon_1 \|\Delta K\|_F^2 + \frac{1}{\epsilon_1} \|\Delta L\|_F^2 \right) \\
&= C_{K1} \|\Delta K\|_F + C_{L1} \|\Delta L\|_F + C_{K2} \|\Delta K\|_F^2 + C_{L2} \|\Delta L\|_F^2 , 
\end{aligned}
$$
where $\epsilon_1 > 0$, $C_{K1} := 2 \|\bar{F}\|_2 \|E_K^{\ddagger}$, $C_{L1} := 2 \|\hat{C}\|_2 \|E_L^{\ddagger}$, $C_{K2} := \|\bar{F}\|_2^2 \|R + \hat{B}^\top S^{\ddagger} \hat{B}\|_2 + 2 \epsilon_1 \|\bar{F}\|_2 \|\hat{C}\|_2 \|\hat{B}^\top S^{\ddagger} \hat{F}^\top\|_2$, and $C_{L2} := 2 \|\hat{C}\|_2 \|E_L^{\ddagger} \|_2 + \frac{2}{\epsilon_1} \|\bar{F}\|_2 \|\hat{C}\|_2 \|\hat{B}^\top S^{\ddagger} \hat{F}^\top\|_2$. 

Then, we expand the terms in the cost difference. 

The cost difference can be further expressed as: 
$$
J(K,L) - J(K^{\ddagger}, L^{\ddagger}) = \text{Tr}(\mathcal{Z} \Omega_{K,L}) = \text{Tr}((\mathcal{L}_K^{\ddagger} + \mathcal{L}_L^{\ddagger} + \mathcal{Q}_{K,L}^{\ddagger}) \Omega_{K,L}).
$$
The stationarity conditions are $\nabla_K J(K^{\ddagger}, L^{\ddagger}) = 0$ and $\nabla_L J(K^{\ddagger}, L^{\ddagger}) = 0$, i.e., 
$$
\nabla_K J(K^{\ddagger}, L^{\ddagger}) = 2 \left( R K^{\ddagger} \bar{F} - \hat{B}^\top S^{\ddagger} \hat{\mathcal{A}}^{\ddagger} \right) \Omega^{\ddagger} \bar{F}^\top = 2 E_K^{\ddagger} \Omega^{\ddagger} \bar{F}^\top = 0, 
$$
$$
\nabla_L J(K^{\ddagger}, L^{\ddagger}) = 2 \hat{F} S^{\ddagger} \hat{\mathcal{A}}^{\ddagger} \Omega^{\ddagger} \hat{C}^\top = 2 E_L^{\ddagger} \Omega^{\ddagger} \hat{C}^\top = 0. 
$$
Since $\Omega^{\ddagger}\bar{F}^\top$ and $\Omega^{\ddagger}\hat{C}^\top$ may not have sufficient rank, we cannot conclude $E_K^{\ddagger} = 0$ or $E_L^{\ddagger} = 0$. 
However, we can evaluate the trace of the linear and quadratic terms against $\Omega_{K,L}$.

Consider the trace of the $\mathcal{L}_K^{\ddagger}$ term: 
$$
\text{Tr}(\mathcal{L}_K^{\ddagger} \Omega_{K,L}) = 2 \text{Tr}(\bar{F}^\top \Delta K^\top E_K^{\ddagger} \Omega_{K,L}) = 2 \text{Tr}(\Delta K^\top E_K^{\ddagger} \Omega_{K,L} \bar{F}^\top). 
$$
It can be derived that 
$$
\begin{aligned}
2 \text{Tr}(\Delta K^\top E_K^{\ddagger} \Omega_{K,L} \bar{F}^\top) &= 2 \text{Tr}(\Delta K^\top E_K \Omega_{K,L} \bar{F}^\top) - 2 \text{Tr}(\Delta K^\top \Delta E_K \Omega_{K,L} \bar{F}^\top) \\
&= \text{Tr}(\Delta K^\top \nabla_K J(K, L)) - 2 \text{Tr}(\Delta K^\top (R + \hat{B}^\top S^{\ddagger} \hat{B}) \Delta K \bar{F} \Omega_{K,L} \bar{F}^\top) \\
&\quad + 2 \text{Tr}(\Delta K^\top \hat{B}^\top S^{\ddagger} \hat{F}^\top \Delta L \hat{C} \Omega_{K,L} \bar{F}^\top) + 2 \text{Tr}(\Delta K^\top \hat{B}^\top \Delta S \hat{\mathcal{A}}_{K,L} \Omega_{K,L} \bar{F}^\top) . 
\end{aligned}
$$

Consider the trace of the $\mathcal{L}_L^{\ddagger}$ term: 
$$
\text{Tr}(\mathcal{L}_L^{\ddagger} \Omega_{K,L}) = 2 \text{Tr}(\hat{C}^\top \Delta L^\top E_L^{\ddagger} \Omega_{K,L}) = 2 \text{Tr}(\Delta L^\top E_L^{\ddagger} \Omega_{K,L} \hat{C}^\top). 
$$
It can be derived that 
$$
\begin{aligned}
2 \text{Tr}(\Delta L^\top E_L^{\ddagger} \Omega_{K,L} \hat{C}^\top) &= 2 \text{Tr}(\Delta L^\top E_L \Omega_{K,L} \hat{C}^\top)  - 2 \text{Tr}(\Delta L^\top \Delta E_L \Omega_{K,L} \hat{C}^\top) \\
&= \text{Tr}(\Delta L^\top \nabla_L J(K, L))  - 2 \text{Tr}(\Delta L^\top\hat{F} S^{\ddagger} \hat{F}^\top \Delta L \hat{C} \Omega_{K,L} \hat{C}^\top) \\
&\quad + 2 \text{Tr}(\Delta L^\top \hat{F} S^{\ddagger}\hat{B} \Delta K \bar{F} \Omega_{K,L} \hat{C}^\top) - 2 \text{Tr}(\Delta L^\top \hat{F} \Delta S \hat{\mathcal{A}}_{K,L} \Omega_{K,L} \hat{C}^\top) . 
\end{aligned}
$$

Consider the quadratic terms: 
$$
\begin{aligned}
\text{Tr}(\mathcal{Q}_{K,L}^{\ddagger} \Omega_{K,L}) &= \text{Tr}(\bar{F}^\top \Delta K^\top (R + \hat{B}^\top S^{\ddagger} \hat{B}) \Delta K \bar{F} \Omega_{K,L}) + \text{Tr}(\hat{C}^\top \Delta L^\top \hat{F} S^{\ddagger} \hat{F}^\top \Delta L \hat{C} \Omega_{K,L}) \\
&\quad - \text{Tr}(\bar{F}^\top \Delta K^\top \hat{B}^\top S^{\ddagger} \hat{F}^\top \Delta L \hat{C} \Omega_{K,L}) - \text{Tr}(\hat{C}^\top \Delta L^\top \hat{F} S^{\ddagger} \hat{B} \Delta K \bar{F} \Omega_{K,L}) \\
&= \text{Tr}(\Delta K^\top (R + \hat{B}^\top S^{\ddagger} \hat{B}) \Delta K \bar{F} \Omega_{K,L} \bar{F}^\top) + \text{Tr}(\Delta L^\top \hat{F} S^{\ddagger} \hat{F}^\top \Delta L \hat{C} \Omega_{K,L} \hat{C}^\top) \\
&\quad - 2 \text{Tr}(\Delta K^\top \hat{B}^\top S^{\ddagger} \hat{F}^\top \Delta L \hat{C} \Omega_{K,L} \bar{F}^\top). 
\end{aligned}
$$
Therefore, we have 
$$
\begin{aligned}
J(K,L) - J(K^{\ddagger}, L^{\ddagger}) &= \text{Tr}(\Delta K^\top (R + \hat{B}^\top S^{\ddagger} \hat{B}) \Delta K \bar{F} \Omega_{K,L} \bar{F}^\top) + \text{Tr}(\Delta L^\top \hat{F} S^{\ddagger} \hat{F}^\top \Delta L \hat{C} \Omega_{K,L} \hat{C}^\top) \\
&\quad - 2 \text{Tr}(\Delta K^\top \hat{B}^\top S^{\ddagger} \hat{F}^\top \Delta L \hat{C} \Omega_{K,L} \bar{F}^\top) \\
&\quad + \text{Tr}(\Delta K^\top \nabla_K J(K, L)) - 2 \text{Tr}(\Delta K^\top (R + \hat{B}^\top S^{\ddagger} \hat{B}) \Delta K \bar{F} \Omega_{K,L} \bar{F}^\top) \\
&\quad + 2 \text{Tr}(\Delta K^\top \hat{B}^\top S^{\ddagger} \hat{F}^\top \Delta L \hat{C} \Omega_{K,L} \bar{F}^\top) + 2 \text{Tr}(\Delta K^\top \hat{B}^\top \Delta S \hat{\mathcal{A}}_{K,L} \Omega_{K,L} \bar{F}^\top) \\
&\quad + \text{Tr}(\Delta L^\top \nabla_L J(K, L))  - 2 \text{Tr}(\Delta L^\top\hat{F} S^{\ddagger} \hat{F}^\top \Delta L \hat{C} \Omega_{K,L} \hat{C}^\top) \\
&\quad + 2 \text{Tr}(\Delta L^\top \hat{F} S^{\ddagger}\hat{B} \Delta K \bar{F} \Omega_{K,L} \hat{C}^\top) - 2 \text{Tr}(\Delta L^\top \hat{F} \Delta S \hat{\mathcal{A}}_{K,L} \Omega_{K,L} \hat{C}^\top) \\
&= \text{Tr}(\Delta K^\top \nabla_K J(K, L)) + \text{Tr}(\Delta L^\top \nabla_L J(K, L)) \\
&\quad - \text{Tr}(\Delta K^\top (R + \hat{B}^\top S^{\ddagger} \hat{B}) \Delta K \bar{F} \Omega_{K,L} \bar{F}^\top) - \text{Tr}(\Delta L^\top \hat{F} S^{\ddagger} \hat{F}^\top \Delta L \hat{C} \Omega_{K,L} \hat{C}^\top) \\
&\quad + 2 \text{Tr}(\Delta K^\top \hat{B}^\top S^{\ddagger} \hat{F}^\top \Delta L \hat{C} \Omega_{K,L} \bar{F}^\top) \\
&\quad + 2 \text{Tr}(\Delta K^\top \hat{B}^\top \Delta S \hat{\mathcal{A}}_{K,L} \Omega_{K,L} \bar{F}^\top) - 2 \text{Tr}(\Delta L^\top \hat{F} \Delta S \hat{\mathcal{A}}_{K,L} \Omega_{K,L} \hat{C}^\top) \\
&= \text{Tr}(\Delta K^\top \nabla_K J(K, L)) + \text{Tr}(\Delta L^\top \nabla_L J(K, L)) \\
&\quad - \text{Tr}(\mathcal{Q}_{K,L}^{\ddagger} \Omega_{K,L}) - 2 \text{Tr}(\Delta S \hat{\mathcal{A}}_{K,L} \Omega_{K,L} \Delta \hat{\mathcal{A}}^{\top}) . 
\end{aligned}
$$

Note that $\Omega_{K,L} \succ 0$, $\bar{F}$ and $\hat{C}$ have rows of full rank, one has $\bar{F} \Omega_{K,L} \bar{F}^{\top} \succ 0$ and $\hat{C} \Omega_{K,L} \hat{C}^{\top} \succ 0$. By utilizing the expression of $\mathcal{Q}_{K,L}^{\ddagger}$, the second to last matrix trace term in the cost difference is 
$$
\begin{aligned}
\text{Tr}(\mathcal{Q}_{K,L}^{\ddagger} \Omega_{K,L}) &\ge \text{Tr}(\Delta K^\top (R + \hat{B}^\top S^{\ddagger} \hat{B}) \Delta K \bar{F} \Omega_{K,L} \bar{F}^\top) + \text{Tr}(\Delta L^\top \hat{F} S^{\ddagger} \hat{F}^\top \Delta L \hat{C} \Omega_{K,L} \hat{C}^\top) \\
&\quad - 2 \|\hat{B}^\top S^{\ddagger} \hat{F}^\top\|_2 \|\hat{C} \Omega_{K,L} \bar{F}^\top\|_2 \|\Delta K\|_F \|\Delta L\|_F \\
&\ge \lambda_{\min}(R + \hat{B}^\top S^{\ddagger} \hat{B}) \lambda_{\min}(\bar{F} \Omega_{K,L} \bar{F}^\top) \|\Delta K\|_F^2 + \lambda_{\min}(\hat{F} S^{\ddagger} \hat{F}^\top) \lambda_{\min}(\hat{C} \Omega_{K,L} \hat{C}^\top) \|\Delta L\|_F^2 \\
&\quad - \|\hat{B}^\top S^{\ddagger} \hat{F}^\top\|_2 \|\hat{C} \Omega_{K,L} \bar{F}^\top\|_2 \left( \epsilon_2 \|\Delta K\|_F^2 + \frac{1}{\epsilon_2} \|\Delta L\|_F^2 \right) \\
&= \alpha_{K2} \|\Delta K\|_F^2 + \alpha_{L2} \|\Delta L\|_F^2 , 
\end{aligned}
$$
where $\epsilon_2 > 0$, $\alpha_{K2} := \lambda_{\min}(R + \hat{B}^\top S^{\ddagger} \hat{B}) \lambda_{\min}(\bar{F} \Omega_{K,L} \bar{F}^\top) - \epsilon_2 \|\hat{B}^\top S^{\ddagger} \hat{F}^\top\|_2 \|\hat{C} \Omega_{K,L} \bar{F}^\top\|_2$ and $\alpha_{L2} := \lambda_{\min}(\hat{F} S^{\ddagger} \hat{F}^\top) \lambda_{\min}(\hat{C} \Omega_{K,L} \hat{C}^\top) - \frac{1}{\epsilon_2} \|\hat{B}^\top S^{\ddagger} \hat{F}^\top\|_2 \|\hat{C} \Omega_{K,L} \bar{F}^\top\|_2$. 

Assuming that the spectral norm of the closed-loop matrix $\hat{\mathcal{A}}_{K,L}$ is uniformly bounded, i.e., $\|\hat{\mathcal{A}}_{K,L}\|_2 \le \gamma < 1$. According to the expression of $\Delta S$ and the cumulative state correlation $\Omega_{K,L}$, we have 
$$
\|\Omega_{K,L}\|_2 \le \|Y\|_2 \sum_{t=0}^{\infty} \gamma^{2t} = \frac{\|Y\|_2}{1 - \gamma^2}, 
$$
$$
\|\Delta S\|_F \le \|\mathcal{Z}\|_F \sum_{t=0}^{\infty} \gamma^{2t} = \frac{\|\mathcal{Z}\|_F}{1 - \gamma^2}. 
$$

According to Cauchy-Schwarz inequality, the last matrix trace term in the cost difference is 
$$
\begin{aligned}
&\quad - 2 \text{Tr}(\Delta S \hat{\mathcal{A}}_{K,L} \Omega_{K,L} \Delta \hat{\mathcal{A}}^{\top}) \\
&\le 2 \|\Delta S \hat{\mathcal{A}}_{K,L}\|_F \times \|\Delta \hat{\mathcal{A}} \Omega_{K,L}\|_F \\
&\le 2 \|\Delta S\|_F \|\hat{\mathcal{A}}_{K,L}\|_2 \|\Delta \hat{\mathcal{A}}\|_F \|\Omega_{K,L}\|_2 \\
&\le \frac{2 \gamma \|Y\|_2}{\left(1 - \gamma^2\right)^2} \left(C_{K1} \|\Delta K\|_F + C_{L1} \|\Delta L\|_F + C_{K2} \|\Delta K\|_F^2 + C_{L2} \|\Delta L\|_F^2\right) \\
&\quad \times \left(\|\hat{B}\|_2 \|\bar{F}\|_2 \|\Delta K\|_F + \|\hat{F}\|_2 \|\hat{C}\|_2 \|\Delta L\|_F\right) \\
&= \frac{2 \gamma C_{K1} \|Y\|_2 \|\hat{B}\|_2 \|\bar{F}\|_2}{\left(1 - \gamma^2\right)^2} \|\Delta K\|_F^2 + \frac{2 \gamma C_{L1} \|Y\|_2 \|\hat{F}\|_2 \|\hat{C}\|_2}{\left(1 - \gamma^2\right)^2} \|\Delta L\|_F^2 \\
&\quad + \frac{2 \gamma C_{K1} \|Y\|_2 \|\hat{F}\|_2 \|\hat{C}\|_2}{\left(1 - \gamma^2\right)^2} \|\Delta K\|_F \|\Delta L\|_F + \frac{2 \gamma C_{L1} \|Y\|_2 \|\hat{B}\|_2 \|\bar{F}\|_2}{\left(1 - \gamma^2\right)^2} \|\Delta K\|_F \|\Delta L\|_F \\
&\quad + \frac{2 \gamma C_{K2} \|Y\|_2 \|\hat{B}\|_2 \|\bar{F}\|_2}{\left(1 - \gamma^2\right)^2} \|\Delta K\|_F^3 + \frac{2 \gamma C_{L2} \|Y\|_2 \|\hat{F}\|_2 \|\hat{C}\|_2}{\left(1 - \gamma^2\right)^2} \|\Delta L\|_F^3 \\
&\quad + \frac{2 \gamma C_{K2} \|Y\|_2 \|\hat{F}\|_2 \|\hat{C}\|_2}{\left(1 - \gamma^2\right)^2} \|\Delta K\|_F^2 \|\Delta L\|_F + \frac{2 \gamma C_{L2} \|Y\|_2 \|\hat{B}\|_2 \|\bar{F}\|_2}{\left(1 - \gamma^2\right)^2} \|\Delta L\|_F^2 \|\Delta K\|_F \\
&\le \frac{2 \gamma C_{K1} \|Y\|_2 \|\hat{B}\|_2 \|\bar{F}\|_2}{\left(1 - \gamma^2\right)^2} \|\Delta K\|_F^2 + \frac{2 \gamma C_{L1} \|Y\|_2 \|\hat{F}\|_2 \|\hat{C}\|_2}{\left(1 - \gamma^2\right)^2} \|\Delta L\|_F^2 \\
&\quad + \frac{2 \gamma C_{K2} \|Y\|_2 \|\hat{B}\|_2 \|\bar{F}\|_2}{\left(1 - \gamma^2\right)^2} \|\Delta K\|_F^3 + \frac{2 \gamma C_{L2} \|Y\|_2 \|\hat{F}\|_2 \|\hat{C}\|_2}{\left(1 - \gamma^2\right)^2} \|\Delta L\|_F^3 \\
&\quad + \frac{2 \gamma C_{K1} \|Y\|_2 \|\hat{F}\|_2 \|\hat{C}\|_2}{\left(1 - \gamma^2\right)^2} \left(\frac{\epsilon_3}{2}\|\Delta K\|_F^2 + \frac{1}{2\epsilon_3}\|\Delta L\|_F^2\right) \\
&\quad + \frac{2 \gamma C_{L1} \|Y\|_2 \|\hat{B}\|_2 \|\bar{F}\|_2}{\left(1 - \gamma^2\right)^2} \left(\frac{\epsilon_4}{2}\|\Delta K\|_F^2 + \frac{1}{2\epsilon_4}\|\Delta L\|_F^2\right) \\
&\quad + \frac{2 \gamma C_{K2} \|Y\|_2 \|\hat{F}\|_2 \|\hat{C}\|_2}{\left(1 - \gamma^2\right)^2} \left(\frac{2}{3}\epsilon_5\|\Delta K\|_F^3 + \frac{1}{3\epsilon_5^{1/2}}\|\Delta L\|_F^3\right) \\
&\quad + \frac{2 \gamma C_{L2} \|Y\|_2 \|\hat{B}\|_2 \|\bar{F}\|_2}{\left(1 - \gamma^2\right)^2} \left(\frac{2}{3}\epsilon_6\|\Delta K\|_F^3 + \frac{1}{3\epsilon_6^{1/2}}\|\Delta L\|_F^3\right) \\
&= \beta_{K2} \|\Delta K\|_F^2 + \beta_{K3} \|\Delta K\|_F^3 + \beta_{L2} \|\Delta L\|_F^2 + \beta_{L3} \|\Delta L\|_F^3 , 
\end{aligned}
$$
where $\beta_{K2}$, $\beta_{K3}$, $\beta_{L2}$ and $\beta_{L3}$ are the corresponding coefficients of the quadratic and cubic terms of the Frobenius norm of the error of controller gain and observer gain. Based on the above results, we have 
$$
\begin{aligned}
J(K,L) - J(K^{\ddagger}, L^{\ddagger}) &\le \text{Tr}(\Delta K^\top \nabla_K J(K, L)) - \alpha_{K2} \|\Delta K\|_F^2 + \beta_{K2} \|\Delta K\|_F^2 + \beta_{K3} \|\Delta K\|_F^3 \\
&\quad + \text{Tr}(\Delta L^\top \nabla_L J(K, L)) - \alpha_{L2} \|\Delta L\|_F^2 + \beta_{L2} \|\Delta L\|_F^2 + \beta_{L3} \|\Delta L\|_F^3 . 
\end{aligned}
$$

Finally, we utilized the gradient of the cost value with respect to controller gain and observer gain to derive an upper bound on the cost difference. 

Assuming that the system has sufficient stability margin, there exist a set of parameters $\{\epsilon_i\}_{i=1}^6$ such that $\alpha_{K2} - \beta_{K2} > 0$ and $\alpha_{L2} - \beta_{L2} > 0$. Define the local attraction radii of the controller gain $K$ and observer gain $L$ as $r_K := \frac{\alpha_{K2} - \beta_{K2}}{2 \beta_{K3}}$ and $r_L := \frac{\alpha_{L2} - \beta_{L2}}{2 \beta_{L3}}$. When the controller gain $K$ and observer gain $L$ satisfy $\|\Delta K\|_F \le r_K$ and $\|\Delta L\|_F \le r_L$, then
$$
\begin{aligned}
\beta_{K3} \|\Delta K\|_F^3 &\le \frac{\alpha_{K2} - \beta_{K2}}{2} \|\Delta K\|_F^2, \\ 
\beta_{L3} \|\Delta L\|_F^3 &\le \frac{\alpha_{L2} - \beta_{L2}}{2} \|\Delta L\|_F^2. 
\end{aligned}
$$
Therefore, we have the following gradient dominance property: 
$$
\begin{aligned}
J(K,L) - J(K^{\ddagger}, L^{\ddagger}) &\le \|\nabla_K J(K, L)\|_F \|\Delta K\|_F - \frac{\alpha_{K2} - \beta_{K2}}{2} \|\Delta K\|_F^2 \\
&\quad + \|\nabla_L J(K, L)\|_F \|\Delta L\|_F - \frac{\alpha_{L2} - \beta_{L2}}{2} \|\Delta L\|_F^2 \\
&\le \frac{\|\nabla_K J(K, L)\|_F^2}{2 (\alpha_{K2} - \beta_{K2})} + \frac{\alpha_{K2} - \beta_{K2}}{2} \|\Delta K\|_F^2 - \frac{\alpha_{K2} - \beta_{K2}}{2} \|\Delta K\|_F^2 \\
&\quad + \frac{\|\nabla_L J(K, L)\|_F^2}{2 (\alpha_{L2} - \beta_{L2})} + \frac{\alpha_{L2} - \beta_{L2}}{2} - \frac{\alpha_{L2} - \beta_{L2}}{2} \|\Delta L\|_F^2 \\
&= \frac{\|\nabla_K J(K, L)\|_F^2}{2 (\alpha_{K2} - \beta_{K2})} + \frac{\|\nabla_L J(K, L)\|_F^2}{2 (\alpha_{L2} - \beta_{L2})} . 
\end{aligned}
$$
\end{proof}

\end{document}